\documentclass[
final
]{dmtcs-episciences}

\usepackage[utf8]{inputenc}
\usepackage{subfigure}

\usepackage{amsmath,amssymb,amsfonts,amsthm}
\usepackage{color}
\usepackage{enumerate}
\usepackage{algorithm}
\usepackage{algpseudocode}
\usepackage{textcomp}
\usepackage{fontenc} 
\usepackage{tikz}
\usepackage{tkz-berge}
\usepackage{tkz-graph}
\usepackage{mathtools}
\usetikzlibrary{graphs}
\usetikzlibrary{shapes,arrows}
\usetikzlibrary{graphs.standard}
\usetikzlibrary{arrows.meta}

\newcommand{\tw}{w}	

\newtheorem{theorem}{Theorem}[section]
\newtheorem{lemma}[theorem]{Lemma}

\theoremstyle{definition}
\newtheorem{definition}[theorem]{Definition}

\theoremstyle{remark}

\usepackage[round]{natbib}

\author{Meysam Rajaati Bavil Olyaei\affiliationmark{1,2}
	\and Mohammad Reza Hooshmandasl\affiliationmark{1,2}\\
	\and Michael J. Dinneen\affiliationmark{3}
	\and Ali Shakiba\affiliationmark{4}}
\title{On fixed-parameter tractability of the mixed domination problem for graphs with bounded tree-width}
\affiliation{
	Department of Computer Science, Yazd University, Yazd, Iran\\
	The Laboratory of Quantum Information Processing, Yazd University, Yazd, Iran\\
	Department of Computer Science, The University of Auckland, Auckland, New Zealand\\
	Department of Computer Science, Vali-e-Asr University of Rafsanjan, Rafsanjan, Iran}
\keywords{Mixed Domination, Tree decomposition, Tree-width, Fixed-parameter tractable}
\received{2016-12-28}
\revised{2017-5-23}
\accepted{2018-6-28}
\begin{document}
	\publicationdetails{20}{2018}{2}{2}{2615}
	\maketitle
	\begin{abstract}
		A mixed dominating set for a graph $G = (V,E)$ is a set $S\subseteq V \cup E$  such that every element $x \in (V \cup E) \backslash S$ is either adjacent or incident to an element of $S$.
		The mixed domination number of a graph $G$, denoted by $\gamma_m(G)$, is the minimum cardinality of mixed dominating sets of $G$. Any mixed dominating set with the cardinality of $\gamma_m(G)$ is called a minimum mixed dominating set.
		The mixed domination set (MDS) problem is to find a minimum mixed dominating set for a graph $G$ and is known to be an  NP-complete problem. In this paper, we present a novel approach to find all of the mixed dominating sets, called the AMDS problem, of a graph with bounded tree-width $\tw$.  Our new technique of assigning power values to edges and vertices, and combining with dynamic programming, leads to a  fixed-parameter algorithm of time $O(3^{\tw^{2}}\times \tw^2 \times |V|)$. 
		This shows that MDS is fixed-parameter tractable with respect to tree-width.
		In addition, we theoretically improve the proposed algorithm to solve the MDS problem in $O(6^{\tw} \times |V|)$ time.
	\end{abstract}
	
	\section{Introduction}
		The mixed dominating set (MDS) problem was first introduced in 1977 by Alavi et. al.~\cite{alavi1977total}. The MDS problem has many practical applications such as placing phase measurement units in an electric power system~\cite{zhao2011algorithmic}.
		Also, there are variations and generalizations of the MDS such as Roman MDS and signed Roman MDS which 
		were introduced and studied by  Abdollahzadeh et al.~\cite{ahangar2015mixed, ahangar2015signed}.
		
		An edge dominates its endpoints as well as all of its adjacent edges. Also, a vertex dominates all of its neighboring vertices as well as all of its incident edges. Formally, a set $S \subseteq V \cup E$ of vertices and edges of a graph $G=(V, E)$ is called a MDS if every element $x \in (V \cup E) \setminus S$ is dominated by an element of $S$. The mixed domination number of $G$ is the size of the smallest mixed dominating set of $G$ and is denoted by $\gamma_m(G)$. Finding all of the mixed dominating set of a graph is called AMDS problem.
		
		The MDS problem is NP-complete for general graphs~\cite{zhao2011algorithmic}.
		There exist different approaches to solve an NP-complete problem such as approximation, randomization,  heuristics, and parameterization.
		Several approximation algorithms exist for solving the MDS problem such as a 2-factor one by Hatami~\cite{hatami2007approximation}. 
		It is notable that the MDS problem remains NP-complete even for split graphs due to 
		the high tree-width of the input graph~\cite{lan2013mixed};
		however, the MDS problem is polynomial tractable for cacti and trees~\cite{lan2013mixed}.  
		A parallel concept is proposed by Adhar et al.~in~\cite{adhar1994mixed} which  requires $O(n)$ processors in CRCW PRAM model to solve the MDS in $O(\log n)$ time where $ n $ is the number of graph vertices.
		
		The parameterization method is a well-known technique which considers certain parameters on the input
		constant to get a polynomial time algorithm with respect to the size of the input and may contain exponential terms
		with respect to these fixed parameters.  A famous example of such parameters is the tree-width which was
		introduced by Robertson and Seymour in 1984~\cite{robertson1984graph}. The tree-width parameter has proven to
		be a good coping strategy for tackling the intrinsic difficulty for various NP-hard problems on graphs.  
		The tree-width measures the similarity of a graph to a tree. Since most of the algorithms work efficiently on trees, the tree decomposition of a graph can be used to speed up solving some problems on graphs with a small tree-width. 
		Although some problems in graph theory cannot be solved in polynomial time even with respect to some fixed parameter,
		there are many other interesting problems in graph theory which are fixed parameter tractable (FPT).
		To show that a problem is FPT, one existing way is to express the problem in monadic second-order logic; if a
		problem can be modeled in this way, then it is FPT by Courcelle's famous theorem~\cite{courcelle1992monadic,courcelle2015fly}. The reduction technique, which is an extension of a graph
		reduction to another graph of bounded tree-width by Bodlaender (see~\cite{bodlaender2001reduction}), is another
		technique which helps solving problems in linear time with respect to constant tree-width.
		
		Almost all of the algorithmic approaches that consider the input graph of a constant tree-width use the
		dynamic programming paradigm. For example, Chimani proposed an algorithm to solve the Steiner tree problem 
		using dynamic programming~\cite{chimani2012improved}. 
		
		In~\cite{rajaati2016fixed}, we proposed an approach to solve the problem of finding all of the mixed
		dominating sets (AMDS) for a graph $G = (V,E)$  of bounded tree-width $\tw$ which has time complexity 
		$O(3^{\tw^{2}}\times \tw^2 \times |V|)$.
		Our constructive algorithm shows that the MDS is fixed-parameter tractable with respect to tree-width. 
		As defined later, the fundamental idea we use to solve the MDS problem is to assign power values to vertices. 
		Recently, Jain et al.~in~\cite{jain2017m} enhanced the complexity\footnote{The ``big Oh star'' notation
		$O^*(f(\tw))$ indicates the algorithm runs in time $O(f(\tw) n^c)$, where $n$ is the input size, $c$ 
		is a constant independent to the treewidth $\tw$ and $f()$ is an arbitrary function dependent only on $\tw$.}
		of~\cite{rajaati2016fixed} to $O^*(6^{\tw})$.
		Here they showed how to turn any set $S \subseteq V\cup E$ to satisfy $(i)$ the edges in $S$ form a matching, 
		and $(ii)$ the set of endpoints of edges in $S$ is disjoint from the vertices in $S$, to a minimum sized mixed dominating set. 
		In this paper, we also modify our original proposed algorithm of~\cite{rajaati2016fixed} to solve MDS 
		with time complexity $O^*(6^{\tw})$.
		
		The rest of the paper is organized as follows: In Section~\ref{Sec2}, we give necessary notations and
		definitions. In Section~\ref{Sec3}, we define the concept of charging vertices which is a key part of our
		proposed algorithm. Our proposed algorithm that solves the AMDS is presented in Section~\ref{Sec4}. Then, we
		modify this algorithm to solve MDS. 
		In Section~\ref{Sec5}, we formally show the correctness of the proposed algorithm.  
		Finally, a brief conclusion and ideas for future work are discussed in Section~\ref{Sec7}.
		
	\section{Preliminaries}\label{Sec2}
		In this section, we overview the graph theory that is used throughout the paper. In general, the notation
		used below follows~\cite{west2001introduction} and~\cite{haynes1998fundamentals}.
		
		All graphs considered in this paper are undirected and simple, i.e.~no parallel edges or self-loops. 
		Let $G = (V, E)$ be a graph with the vertex set $V$ and the edge set $E$.
		
		For vertex $v\in V$, $N(v)$ denotes the open neighborhood of $v$ and is defined as $N(v) = \{u \in V \mid uv\in E\}$. The edge open  neighborhood of the vertex $v$ is defined as $N^{e}(v) = \{e \in E \mid e=uv\}$. 
		Also, for an edge $e=uv \in E$, $N(e) = \{u,v\}$ denotes the open neighborhood of $e$. 
		The edge open neighborhood of the edge $e=uv$ is defined as $N^{e}(e) = \{e^{\prime} \in E \mid 
		e^{\prime}=uv^{\prime}$ or $e'=u'v$ where $u\not=u'$ and $v\not=v'\}$. 
		We denote the mixed neighborhood of vertex $v$  by $N^{md}(v)$ such that $N^{md}(v) = N(v) \cup N^{e}(v)$, 
		and the mixed neighborhood of edge $e$  by $N^{md}(e)$ such that $N^{md}(e) = N(e) \cup N^{e}(e)$. 
		Finally, for any element $r\in V \cup E$, we denote the mixed neighborhood of $r$ by $N^{md}_{G}(r)$. 
		Also, for any element $r\in V \cup E$, the closed mixed neighborhood is defined as $N^{md}_G(r)\cup \{r\}$ 
		and is denoted by $N^{md}_G[r]$. 
		
		A tree decomposition of a graph $G$ is a mapping of $G$ into a tree $T$ which satisfies certain properties.
		Note that throughout the paper, nodes of $G$ are called vertices while nodes of $T$ are called bags.
		\begin{definition}
			Let $G=(V,E)$ be a graph. A tree decomposition of $G$ is a pair $(\mathcal{X}=\{X_i \mid i\in \mathcal{I}\},T)$, where each $X_i$ is a subset of $V$, which is called a bag, $T$ is a tree with elements of $\mathcal{I}$ as bags and satisfies the following three properties.
			\begin{enumerate}
				\item $\bigcup_{i\in \mathcal{I}}X_i=V$,
				\item for every edge $\{u,v\}\in E$, there is an index $i\in \mathcal{I}$ such that $\{u,v\}\subseteq X_i$,
				\item for all $i,j,k \in \mathcal{I}$, if $j$ lies on the path between $i$ and $k$ in $T$, then $X_i \cap X_k \subseteq X_j$.
			\end{enumerate}
		\end{definition}
		The width of a tree decomposition  $(\mathcal{X}=\{X_i \mid i\in \mathcal{I}\},T)$ equals to $\max\{|X_i|
		\mid i\in \mathcal{I}\}-1$. The tree-width of a graph $G$, denoted by $\tw$,  is the minimum width among all the tree decompositions of graph $G$.
		\begin{definition}
			A tree decomposition  $(\mathcal{X}=\{X_i \mid i\in \mathcal{I}\},T)$  is called a nice tree decomposition if the following conditions are met:
			\begin{enumerate}
				\item Every bag of the tree has at most two children.
				\item If a bag $i$ has two children $j$ and $k$, then $X_i=X_j=X_k$. Such a bag is called a \textsc{join} bag.
				\item If a bag $i$ has exactly one child like $j$, then one of the following conditions must hold:
				\begin{description}
					\item[(a)] $|X_i|=|X_j|+1$ and $X_j\subset X_i$
					\item[(b)] $|X_i|=|X_j|-1$ and $X_i\subset X_j$		
				\end{description}
				Note that if (a) holds, the bag $|X_i|$ is called an \textsc{introduce} bag, and if (b) holds, it is called a \textsc{forget} bag.
			\end{enumerate}
		\end{definition}
		\begin{lemma}\label{lemordertree-width}(\cite{bodlaender1996linear})
			Given a tree decomposition of a graph $G$ of width $\tw$,  and $n$ vertices, one can find a nice tree decomposition of $G$ in linear time of width $\tw$ and $O(n)$ bags. 
		\end{lemma}
		\begin{definition}\label{vntd}
			A nice tree decomposition is called a very nice tree decomposition  if each \textsc{leaf} bag contains just a single vertex.
		\end{definition}
		\begin{figure}[htbp]
			\begin{center}
				\subfigure[\label{maingraph}]\centering
				\resizebox*{.3\textwidth}{!}{			
					\begin{tikzpicture}			
					\node[shape=circle,draw=black] (1) at (0,0) {1};
					\node[shape=circle,draw=black] (2) at (-2,-2) {2};
					\node[shape=circle,draw=black] (3) at (0,-2) {3};
					\node[shape=circle,draw=black] (4) at (2,-2) {4};
					\node[shape=circle,draw=black] (5) at (4,-2) {5};			
					\Edge[label = 1](1)(2);
					\Edge[label = 2](1)(3);
					\Edge[label = 3](1)(4);
					\Edge[label = 4](2)(3);
					\Edge[label = 5](3)(4);
					\Edge[label = 6](4)(5);
					\end{tikzpicture}}
				\hfil
				\subfigure[\label{NTD}]\centering
				\resizebox*{.5\textwidth}{!}{			
					\begin{tikzpicture}			
					\node[shape=circle,draw=black] (1) at (0,0) {1};
					\node[shape=circle,draw=white] (13) at (2,0) {12: \textsc{forget} bag};
					\node[shape=ellipse,draw=black] (2) at (0,-1.25) {1, 4};
					\node[shape=circle,draw=white] (14) at (2,-1.25) {11: \textsc{join} bag};
					\node[shape=ellipse,draw=black] (3) at (-1,-2.5) {1, 4};
					\node[shape=circle,draw=white] (15) at (-3.25,-2.5) {6: \textsc{forget} bag};
					\node[shape=ellipse,draw=black] (4) at (-1,-3.75) {1, 4, 3};
					\node[shape=circle,draw=white] (16) at (-4,-3.75) {5: \textsc{introduce} bag};
					\node[shape=ellipse,draw=black] (5) at (-1,-5) {1, 3};
					\node[shape=circle,draw=white] (17) at (-3.25,-5) {4: \textsc{forget} bag};
					\node[shape=ellipse,draw=black] (6) at (-1,-6.25) {1, 2, 3};
					\node[shape=circle,draw=white] (18) at (-4,-6.25) {3: \textsc{introduce} bag};
					\node[shape=ellipse,draw=black] (7) at (-1,-7.5) {2, 3};
					\node[shape=circle,draw=white] (19) at (-3.25,-7.5) {2: \textsc{forget} bag};
					\node[shape=circle,draw=black] (8) at (-1,-8.75) {2};
					\node[shape=circle,draw=white] (20) at (-3,-8.75) {1: \textsc{leaf} bag};
					\node[shape=ellipse,draw=black] (9) at (1.5,-2.5) {1, 4};
					\node[shape=circle,draw=white] (21) at (4.25,-2.5) {10: \textsc{introduce} bag};
					\node[shape=circle,draw=black] (10) at (1.5,-3.75) {4};
					\node[shape=circle,draw=white] (22) at
					(3.5,-3.75) {9: \textsc{forget} bag}; 
					\node[shape=ellipse,draw=black] (11) at (1.5,-5) {4, 5};
					\node[shape=circle,draw=white] (23) at (4.25,-5) {8: \textsc{introduce} bag};
					\node[shape=circle,draw=black] (12) at (1.5,-6.25) {5};	
					\node[shape=circle,draw=white] (24) at (3.5,-6.25) {7: \textsc{leaf} bag};
					
					\draw[draw=black,line width=1pt] (1) to (2);
					\draw[draw=black,line width=1pt] (2) to (3);
					\draw[draw=black,line width=1pt] (3) to (4);
					\draw[draw=black,line width=1pt] (4) to (5);   
					\draw[draw=black,line width=1pt] (5) to (6);
					\draw[draw=black,line width=1pt] (6) to (7);
					\draw[draw=black,line width=1pt] (7) to (8);
					\draw[draw=black,line width=1pt] (2) to (9);
					\draw[draw=black,line width=1pt] (9) to (10);
					\draw[draw=black,line width=1pt] (10) to (11);  
					\draw[draw=black,line width=1pt] (11) to (12);  
					\end{tikzpicture}}
				\caption{(a) Graph $G1$, (b) One of nice tree decomposition of $G1$ with treewidth 2. The bags of tree decomposition  are numbered according to a preorder traversal on it.}
				\label{PGandTD}
			\end{center}
		\end{figure}
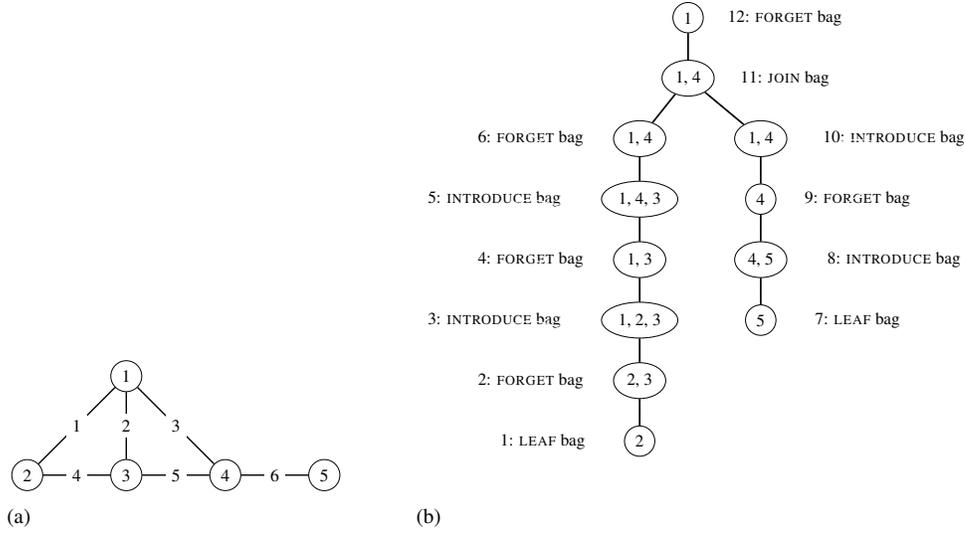
		
	\section{Fundamental Concepts}\label{Sec3}
		The fundamental idea that we use to solve the AMDS is transferring the edge power to the vertex power.  
		
		Let $\mathcal{MD}$ be a mixed domination set, $X_{i}$ be a bag in a tree decomposition of $G$ and $v \in X_i$ be a vertex of $G$. 
		The rules for transferring the domination power of the edges to the vertices are as follows:
		\begin{itemize}
			\item \textbf{Power 2:} If $v \in \mathcal{MD}$, then the power of $v$ is set equal to 2. In this case, this vertex can dominate the elements in $N^{md}_G[r]$.
			\item \textbf{Power 1:} If $v \notin \mathcal{MD}$ and at least one of the incident edges of the vertex $v$ is in $\mathcal{MD}$, then the power of $v$ is set equal to $1$. This vertex can dominate all of its incident edges.
			\item \textbf{Power 0:} If $v$ and all of its incident edges are not in $\mathcal{MD}$, then the power of $v$ is set equal to $0$. This vertex cannot dominate any edges or vertices.
		\end{itemize}
		Given a vertex $v \in X_{i}$, there are seven situations to consider which are illustrated in the Table~\ref{bag condition vertex}. The intuition behind each of these cases is as follows: 
		\begin{enumerate}
			\item In the first case, the vertex $v$ and at least one of its incident edges  belong to $\mathcal{MD}$,
			\item In the second case, $v$ belong to $\mathcal{MD}$, however, none of its edges belong to  $\mathcal{MD}$,
			\item In the third case, $v$ does not belong to $\mathcal{MD}$ but at least one of its incident edges
			belongs to  $\mathcal{MD}$,
			\item In the fourth case, $v$ and all of its edges do not belong to $\mathcal{MD}$, and $v$ and its edges are dominated,
			\item In the fifth case, $v$ and all of its edges do not belong to $\mathcal{MD}$, and $v$ is not dominated but all its edges are dominated,
			\item In the sixth case, $v$ and all of its edges do not belong to $\mathcal{MD}$, and $v$ is dominated but at least one of its edges is not dominated,
			\item In the seventh case, $v$ and all of its edges do not belong to $\mathcal{MD}$, and $v$ and at least one of its edges is not dominated.
		\end{enumerate}
		
		\begin{table}[htbp]
			\centering
			\label{bag condition vertex}
			\caption{The seven possible situation for a vertex in a  bag.}
			\scalebox{0.7}{
			\begin{tabular}{ r|c|c|c|c|c|c|}
						\multicolumn{1}{r}{}
						&  \multicolumn{1}{c}{$v\in \mathcal{MD}$}
						&  \multicolumn{1}{c}{$N^{e}_{X_{i}}(v)\in \mathcal{MD}$ }
						&  \multicolumn{1}{c}{vertex cover}
						& \multicolumn{1}{c}{edge cover} 
						&  \multicolumn{1}{c}{vertex power}
						&  \multicolumn{1}{c}{illustration}
						\\
						\cline{2-7}
						1 & $v \in \mathcal{MD}$ & $\exists e \in N^e_{X_{i}}(v) , e \in \mathcal{MD}$ & $v$ is covered & $\forall e \in N^e_{X_{i}}(v), e$ is covered & 2 &  \ref{l1}.
						\\
						\cline{2-7}
						2 & $v \in \mathcal{MD}$ & $\forall e \in N^e_{X_{i}}(v) , e \notin \mathcal{MD}$  & $v$ is covered & $\forall e \in N^e_{X_{i}}(v), e$ is covered & 2 & \ref{l2}.
						\\
						\cline{2-7}
						3 & $v \notin \mathcal{MD}$ & $\exists e \in N^e_{X_{i}}(v) , e \in \mathcal{MD}$ & $v$ is covered & $\forall e \in N^e_{X_{i}}(v), e$ is covered & 1 & \ref{l3}.
						\\
						\cline{2-7}
						4 & $v \notin \mathcal{MD}$ & $\forall e \in N^e_{X_{i}}(v), e \notin \mathcal{MD}$ & $v$ is covered & $\forall e \in N^e_{X_{i}}(v), e$ is covered & 0 &   \ref{l4}.
						\\
						\cline{2-7}
						5 & $v \notin \mathcal{MD}$ & $\forall e \in N^e_{X_{i}}(v), e \notin \mathcal{MD}$ & $v$ is not covered & $\forall e \in N^e_{X_{i}}(v), e$ is covered & 0 & \ref{l5}.
						\\
						\cline{2-7}
						6 & $v \notin \mathcal{MD}$ & $\forall e \in N^e_{X_{i}}(v), e \notin \mathcal{MD}$ & $v$ is covered & $\exists e \in N^e_{X_{i}}(v), e$ is not covered & 0 &  \ref{l6}.
						\\
						\cline{2-7}
						7 & $v \notin \mathcal{MD}$ & $\forall e \in N^e_{X_{i}}(v), e \notin \mathcal{MD}$ & $v$ is not covered & $\exists e \in N^e_{X_{i}}(v), e$ is not covered & 0 &    \ref{l7}.
						\\
						\cline{2-7}
					\end{tabular}
				}
		\end{table}
		In Figure~\ref{bagpower} we illustrate these situations where a rectangle indicates a bag; 
		the vertices or edges that are in $\mathcal{MD}$ are drawn as disk or bold line, respectively.  
		Not covered elements are drawn by dotted lines or circles and the remaining elements are covered. 
		In Figures~\ref{l2},~\ref{l4} and~\ref{l5}, we use an arc sector for incident edges of the selected vertex 
		in the bag. It means all of the edges are covered, but in Figures~\ref{l1} and~\ref{l3} at least one of
		incident edges of the vertex is in $\mathcal{MD}$ and in Figures~\ref{l6} and~\ref{l7} 
		at least one of incident edges of the vertex is not covered. 
		\begin{figure}[htbp]
			\begin{center}
				\subfigure[ 1\label{l1}]{\includegraphics[height=1.2cm]{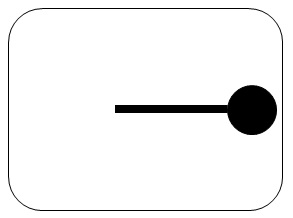}}
				\hfil
				\subfigure[2\label{l2}]{\includegraphics[height=1.2cm]{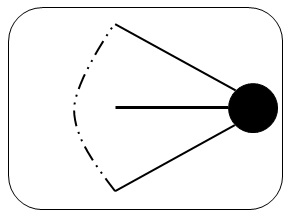}}
				\hfil
				\subfigure[3\label{l3}]{\includegraphics[height=1.2cm]{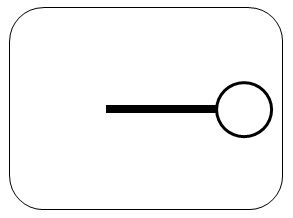}}
				\hfil
				\subfigure[4\label{l4}]{\includegraphics[height=1.2cm]{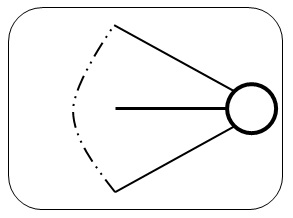}}
				\hfil
				\subfigure[5\label{l5}]{\includegraphics[height=1.2cm]{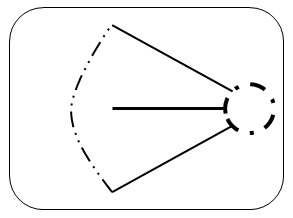}}
				\hfil
				\subfigure[6\label{l6}]{\includegraphics[height=1.2cm]{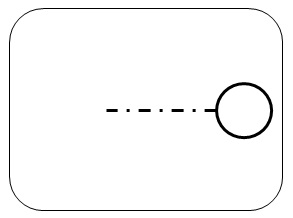}}
				\hfil
				\subfigure[7\label{l7}]{\includegraphics[height=1.2cm]{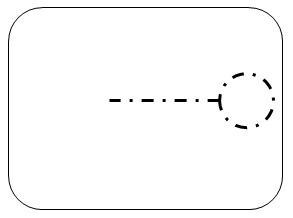}}
				\caption{Different situations for a vertex with respect to a  bag.}
				\label{bagpower}
			\end{center}
		\end{figure}
		Let $T$ be a  very nice tree decomposition (recall Definition~\ref{vntd}) for the graph $G$. For each bag $X_{i}$ in $T$, we define the set $X^{\triangle}_{i} $ as 
		\begin{equation}
		X^{\triangle}_{i}  = \{v \in V \mid v \text{ is in the descendant bags of } X_{i} \mbox{ within } \mathcal{I}\}.
		\end{equation}
		The induced subgraph of $G$ with vertices $X_{i}$ $(\text{or }X^{\triangle}_{i})$ is denoted by  $G_{i}$ $(\text{or }G_{i}^{\triangle})$.
		Let $\mathcal{MD}$ be a mixed dominating set for the bag $X_i$ and $v \in X_{i} $. For a vertex $v$, there
		are nine possible situations to consider based on the intersection of $G_{i}^{\triangle}$ and bag $X_{i}$.
		These are given in Table~\ref{condition vertex}. Assuming that all of the elements in
		$G_{i}^{\triangle}/G_{i}$ are dominated, seven situations of the Table~\ref{condition vertex} are the same as
		earlier given in Table~\ref{bag condition vertex}. However, to cover the cases that edges in $G_{i}^{\triangle}/G_{i}$ are not dominated, two extra cases are possible for $v$.
		\begin{itemize}
			\item[8.] The vertex $v$ is dominated, however at least one of its  edges in $G_{i}^{\triangle}/G_{i}$ is not dominated.
			\item[9.] The vertex $v$ and at least one of its  edges  in $G_{i}^{\triangle}/G_{i}$ are not dominated.
		\end{itemize}
		In both case, the vertex $v$ is not in $\mathcal{MD}$, and at least one of its incident edges in $G_{i}^{\triangle}/G_{i}$ is not dominated. However the vertex $v$ in state 8 is dominated, and in state 9 is not dominated.
		\begin{table}[htbp]
			\label{condition vertex}
			\caption{Different cases for a vertex $v \in X_i$ in correspondence to the edges and vertices in $X_i$.}
			\scalebox{0.65}{
			\begin{tabular}{ r|c|c|c|c|c|c|c| }
						\multicolumn{1}{r}{}
						&  \multicolumn{1}{c}{$v\in \mathcal{MD}$}
						&  \multicolumn{1}{c}{$N^e_{X^{\triangle}_{i}}(v)\in \mathcal{MD}$ }
						&  \multicolumn{1}{c}{vertex cover}
						& \multicolumn{1}{c}{edge is covered in this bag} 
						& \multicolumn{1}{c}{edge is covered in previous bags}
						&  \multicolumn{1}{c}{power}
						&  \multicolumn{1}{c}{illustration }
						\\
						\cline{2-8}
						1 & $v \in \mathcal{MD}$ & $\exists e  \in N^e_{X^{\triangle}_{i}}(v), e \in \mathcal{MD}$ & $v$ is covered & $\forall e \in N^e_{X_{i}}(v), e $ is covered & $\forall e \in N^e_{X^{\triangle}_{i}}(v), e $ is covered & 2 &  Figure \ref{l11}
						\\
						\cline{2-8}
						2 & $v \in \mathcal{MD}$ & $\forall e \in N^e_{X^{\triangle}_{i}}(v), e \notin \mathcal{MD}$  & $v$ is  covered & $\forall e \in N^e_{X_{i}}(v), e $ is covered &$\forall e \in N^e_{X^{\triangle}_{i}}(v), e $ is covered & 2 & Figure \ref{l12}
						\\
						\cline{2-8}
						3 & $v \notin \mathcal{MD}$ & $\exists e  \in N^e_{X^{\triangle}_{i}}(v), e \in \mathcal{MD}$ & $v$ is  covered & $\forall e \in N^e_{X_{i}}(v), e$ is covered &$\forall e \in N^e_{X^{\triangle}_{i}}(v), e$ is covered & 1 & Figure \ref{l13}
						\\
						\cline{2-8}
						4 & $v \notin \mathcal{MD}$ & $\forall e \in N^e_{X^{\triangle}_{i}}(v), e \notin \mathcal{MD}$ & $v$ is  covered & $\forall e \in N^e_{X_{i}}(v), e$ is covered & $\forall e \in N^e_{X^{\triangle}_{i}}(v), e$ is covered & 0 &   Figure \ref{l14}
						\\
						\cline{2-8}
						5 & $v \notin \mathcal{MD}$ & $\forall e \in N^e_{X^{\triangle}_{i}}(v), e \notin \mathcal{MD}$ & $v$ is not covered & $\forall e \in N^e_{X_{i}}(v), e$ is covered &$\forall e \in N^e_{X^{\triangle}_{i}}(v), e$ is covered & 0 & Figure \ref{l15}
						\\
						\cline{2-8}
						6 & $v \notin \mathcal{MD}$ & $\forall e \in N^e_{X^{\triangle}_{i}}(v), e \notin \mathcal{MD}$ & $v$ is  covered & $\exists e  \in N^e_{X_{i}}(v), e$ is not covered &$\forall e \in N^e_{X^{\triangle}_{i}}(v), e$ is covered & 0 &  Figure \ref{l16}
						\\
						\cline{2-8}
						7 & $v \notin \mathcal{MD}$ & $\forall e \in N^e_{X^{\triangle}_{i}}(v), e \notin \mathcal{MD}$ & $v$ is not covered & $\exists e  \in N^e_{X_{i}}(v), e$ is not covered &$\forall e \in N^e_{X^{\triangle}_{i}}(v), e$ is covered & 0 &    Figure \ref{l17}
						\\
						\cline{2-8}
						8 & $v \notin \mathcal{MD}$ & $\forall e \in N^e_{X^{\triangle}_{i}}(v), e \notin \mathcal{MD}$ &  $v$ is  covered & $ e $ is or is not covered &$\exists e  \in N^e_{X^{\triangle}_{i}}(v), e$ is not covered & 0 &  Figure \ref{l18}
						\\
						\cline{2-8}
						9 & $v \notin \mathcal{MD}$ & $\forall e \in N^e_{X^{\triangle}_{i}}(v), e \notin \mathcal{MD}$ & $v$ is not covered  & $ e $ is or is not covered &$\exists e  \in N^e_{X^{\triangle}_{i}}(v), e $ is not covered & 0 &  Figure \ref{l19}
						\\
						\cline{2-8}
					\end{tabular}
					}
		\end{table}
		The Figure~\ref{prepower} is similar to the Figure~\ref{bagpower}, except that it  shows different
		situations for a vertex $v \in X_i$ with respect to a bag $X_i$, and it considers the edges and vertices
		appearing in previous bags.
		\begin{figure}[htbp]
			\begin{center}
				\subfigure[ 1\label{l11}]{\includegraphics[height=.6cm]{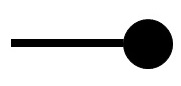}}
				\hfil
				\subfigure[2\label{l12}]{\includegraphics[height=1.2cm]{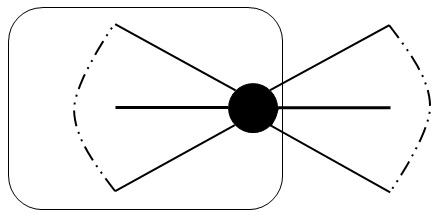}}
				\hfil
				\subfigure[3\label{l13}]{\includegraphics[height=.6cm]{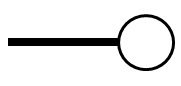}}
				\hfil
				\subfigure[4\label{l14}]{\includegraphics[height=1.2cm]{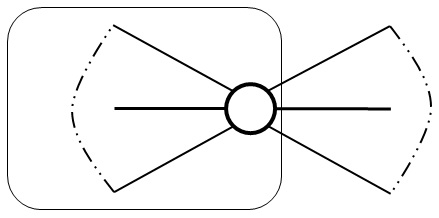}}
				\hfil
				\subfigure[5\label{l15}]{\includegraphics[height=1.2cm]{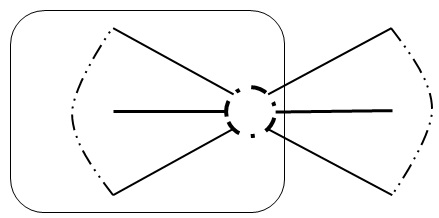}}
				\hfil\\
				\subfigure[6\label{l16}]{\includegraphics[height=1.2cm]{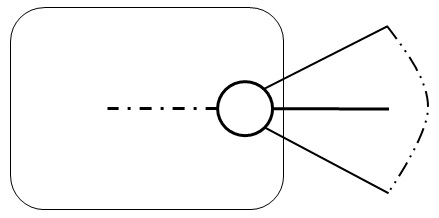}}
				\hfil
				\subfigure[7\label{l17}]{\includegraphics[height=1.2cm]{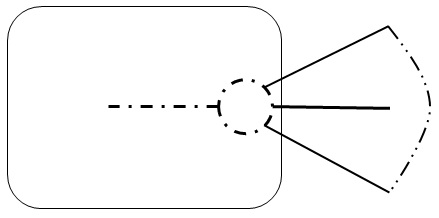}}
				\hfil
				\subfigure[8\label{l18}]{\includegraphics[height=1.2cm]{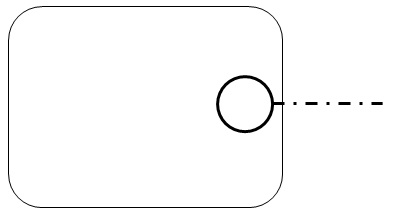}}
				\hfil
				\subfigure[9\label{l19}]{\includegraphics[height=1.2cm]{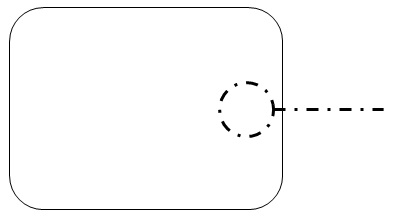}}
				\caption{Different cases for a vertex $v \in X_i$ in correspondence to the edges and vertices in $X_i$.}
				\label{prepower}
			\end{center}
		\end{figure}
		For each edge, there are three different possible cases to consider: (1) it is in the mixed dominating
		set, (2) it is not in the mixed dominating set and it is covered by some element(s), or (3) it is not in  the mixed dominating set and  is  not covered by any element.
		These cases are summarized in Table~\ref{Possibleconditionforanedge}.
		\begin{table}[htpb]	
			\label{Possibleconditionforanedge}
			\caption{Possible condition for an edge.}
			\centering
			\begin{tabular}{ r|c|}
						\cline{2-2}
						1 & edge belongs to the mixed dominating set.   \\
						\cline{2-2}
						2 & edge does not belong to mixed dominating set, but it is covered     \\
						\cline{2-2}
						3 & edge is not in mixed dominating set and also is not covered.    \\
						\cline{2-2}
					\end{tabular}
		\end{table}
		For our algorithm, we keep two types of data tables:  (1) The bag table $Btable_i$ saves all of the possible states 
		for the vertices and the edges in the bag $X_i$. (2) The status table $Stable_i$ saves all possible states 
		for the vertices and the edges in the bag $X_i$ with respect to $G_{i}^{\triangle}$. 
		Each table is constructed as follows: in both $Btable_{i}$ and $Stable_{i}$ each row represents a 
		possible solution and each column corresponds to a vertex or an edge in the induced graph $X_i$. 
		In addition they both have a column with a cost value that shows the number of mixed domination members in that row. 
		The ordering of cells in each row of these tables is $v_{0}, \; ..., \; v_{\tw}, \; e_{1}, \; ..., \; e_{\binom{\tw+1}{2} }, \; cost$.
		
	\section{Our Proposed  Algorithm}\label{Sec4}
		In this section, we present our proposal algorithm to find the mixed domination number for a graph with 
		bounded tree-width. This algorithm consists of three phases:
		\begin{itemize}
			\item[Step 1:] 
				Let $G=(V,E)$ be an unweighted and undirected graph with constant tree-width. We compute and then use a
				standard very nice tree decomposition with width $\tw$. It can be done in time $O(n)$ using  Lemma~\ref{lemordertree-width}.
			\item[Step 2:]  We find a postorder traversal $\tau$ on the very nice tree decomposition.
				The traversal $\tau$ begins from the leftmost leaf and then goes up in the tree until it reaches the first \textsc{join} bag. Then, it goes to the leftmost leaf on the right subtree of the \textsc{join} bag recursively. It goes up if both children of the \textsc{join} bag are visited and visits the \textsc{join} bag itself and continues until it reaches the root. This phase can be computed in $O(n)$ time.
			\item[Step 3:] We follow elements of $\tau$ in order and update the corresponding tables for each bag as follows:
				\begin{itemize}
					\item Whenever we reach a \textsc{leaf} bag, we create a new table which contains all of the possible 
					cases that the bag can be in.
					\item Whenever we reach an \textsc{introduce} bag, we construct a new table for the bag from its child table.
					\item Whenever we reach a \textsc{join} bag, we construct a new table for the bag from its children tables.
					\item Whenever we reach a \textsc{forget} bag, we construct the table to obtain all of the possible states that vertices in this bag can have by considering edges and vertices which appeared heretofore.
				\end{itemize} 
			It is clear that the time complexity of this phase relates to the time spent in processing each bag 
				in the traversal $\tau$. This process includes time to create a table for each \textsc{leaf} bag and 
				time to combine two tables for the \textsc{introduce}, the \textsc{forget} and the \textsc{join} bags. A  new table is created with two rows, running in constant time, therefore these operations add a constant time factor. 
				To combine two tables we consider the worst case. 
				The table for a \textsc{join} or an \textsc{introduce} bag can be created in $O(\tw^{2})$ steps and for a 
				\textsc{forget} bag requires $O(\tw)$ steps. 
				The worst case for \textsc{join} bag happens when they have all possible cases which lead to
				$9^{\tw+1}\times 3^{\binom{\tw+1}{2}}$ rows. Therefore, the  time complexity of this phase equals
				$O(9^{\tw}\times 3^{\tw^{2}}\times \tw^{2})$. 
		\end{itemize}
		The algorithm described so far is polynomial-time with respect to the size of $T$, 
		Lemma~\ref{lemordertree-width} shows this size is $O(n)$. 
		However, it is exponential with respect to the tree-width of $T$, $\tw$. 
		The following theorem states the time complexity of our proposed algorithm. 
		\begin{theorem} \label{timecom}
			The running time of the described approach is $O(9^{\tw}\times 3^{\tw^{2}}\times \tw^{2}\times n)$.
		\end{theorem}
		Our proposed dynamic programming algorithm works on $\tau$ which is the postorder traversal on a very nice tree decomposition of $G$. When the algorithm visits a bag, it describes the partial solutions to AMDS and as it continues to other bags, it extends the created partial solutions. These partial solutions need to satisfy all of the problem specific constraints in $G_{i}^{\triangle}$ except for the vertices in $X_{i}$ and their incident edges in $G_{i}^{\triangle} \backslash G_{i}$. The status tables  $Stable$ are used to store these partial solutions. In other words, a $Stable$ characterizes the partial solutions and each row in $Stable$ contains a valid assignment for vertices and edges in $G_{i}$.
		
		To compute the $Stable$ of each bag $i$, our algorithm uses $Btable$ of the children of bag $i$. Since these
		tables are computed bottom-up, the final solution of the MDS appears in the root of $T$. So, it can be
		extracted by inspecting the table of the root. The Algorithm~\ref{MixedAlgorithm}  demonstrates how we achieve the final answer $\gamma_{md}$.

		\begin{algorithm}	 
			\caption{ The algorithm to compute $\gamma_{md}$. }
			\begin{algorithmic}[h!]
				\\ INPUT: Postorder traversal $\tau$ on a standard very nice tree decomposition of graph $G$. 
				\\ OUTPUT: $\gamma_{md}$ for $G$.
				\For{$ i \gets 1 $ to    $ |\tau|$}
				\If{ $X_i$ is a bag leaf}
				\State  Create a new table with two rows each corresponding to cases 2 and 5 for the isolated
				vertex in $X_i$ (see Section~\ref{Create leaf}) .
				\ElsIf { $X_i$ is an \textsc{introduce} bag}
				\For{$ \ell_{1} \gets 1 $ to  number of rows in  $Stable_{i-1}$}
				\For{$ \ell_{2} \gets 1 $ to  number of rows in  $Btable_{i}$}
				\State 
				Call Algorithm~\ref{Introducealgorithm} with inputs $r_{Stable_{i-1}}(\ell_{1},:)$ and $ r_{Btable_{i}}(\ell_{2},:)$. 
				\EndFor
				\EndFor
				\ElsIf { $X_i$ is a \textsc{forget} bag }
				\For{$ \ell_{1} \gets 1 $ to   number of rows in  $Stable_{i-1}$}
				\State 
				Call Algorithm~\ref{forgetal} with input $r_{Stable_{i-1}}(\ell_{1},:)$.
				\EndFor
				\ElsIf { $X_i$ is a \textsc{join} bag  }
				\For{$ \ell_{1} \gets 1 $ to  number of rows in  $Stable_{i1}$}
				\For{$ \ell_{2} \gets 1 $ to  number of rows in  $Stable_{i2}$}
				\State Call Algorithm~\ref{joinalgorithm} with inputs $r_{Stable_{i1}}(\ell_{1},:)$ and $ r_{Stable_{i2}}(\ell_{2},:)$.
				\EndFor
				\EndFor
				\EndIf  
				\State Add the  created rows to $Stable_{i}$.
				\EndFor
				\end{algorithmic}\label{MixedAlgorithm}
		\end{algorithm} 
		
		When Algorithm~\ref{MixedAlgorithm} observes a \textsc{leaf} bag, it creates a new table which saves 
		all of the possible states for the only vertex in that bag. 
		The algorithm when following the traversal $\tau$  calls Algorithms~\ref{Introducealgorithm},~\ref{forgetal}
		and~\ref{joinalgorithm} when observing introduce, \textsc{forget} and \textsc{join} bags respectively, and return $\gamma_{md}$ as output. Note that it is possible that combining two rows in different levels 
		may create the same rows. In this case we store the row with minimum cost in $ Stable_{i}$.  
		To avoid searching and sorting to find these repeated states, we use a coding to store the created rows. 
		We use a help table in which each created row has a specific position in it. When a row is created while
		combining two tables, the help table is checked and if there is a row with lower cost, then the lower cost is
		considered and the help table is updated accordingly. Finally  Algorithm~\ref{MixedAlgorithm} inspects the root table and finds $\gamma_{md}$. 
		For illustration, we consider the graph in Figure~\ref{maingraph} and one of its nice tree decomposition 
		(see Figure~\ref{NTD}) as an example for our proposed algorithm. 
		The output of the algorithm is  $\gamma_{md}=2$. Next, we describe how the tables are filled and partial solutions are computed.
		\begin{figure}[htbp]
			\begin{center}
				\includegraphics[width=0.35 \textwidth]{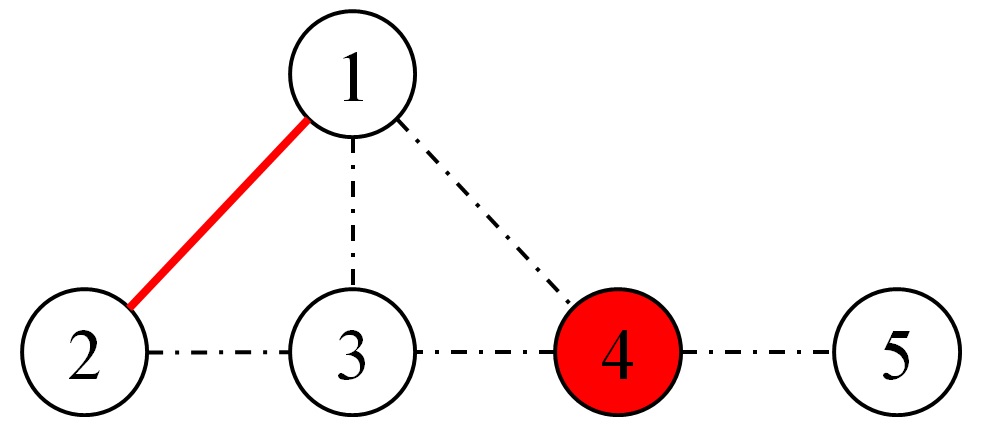}
				\caption{Mixed domination of the graph $G1$.}
				\label{mixeddomination}
			\end{center}
		\end{figure}
		
		\subsection{Status table construction for \textsc{leaf} bags}\label{Create leaf}
		During the traversal $\tau$, when we observe a \textsc{leaf} bag, a new table is created which saves all of the possible states for the only vertex in that bag.
		Let $X_{i} $  be a \textsc{leaf} bag since we are working with a standard nice tree decomposition, every
		\textsc{leaf} bag contains exactly one vertex. For vertex $v \in X_{i} $, the status table $i$ denoted as
		$Stable_{i}$ contains two rows. The first row corresponds to the situation in with vertex $v$ belongs to an
		optimal mixed domination set, and the last one is for the case in which vertex $v$ does not belong to an
		optimal mixed domination set. Therefore, this table has two rows wherein the first row, the value of the
		vertex is $ ``2" $   by costing $ ``1" $, and in the second row, the value of the vertex is $ ``5" $   by
		costing $ ``0" $. The status table of \textsc{leaf} bag 1 in Figure~\ref{NTD} is shown in Table~\ref{leaftable}.
		\begin{table}
			\label{leaftable}
			\caption{The status table of \textsc{leaf} bag 1 in Figure~\ref{NTD}.}
			\centering
			\begin{tabular}{|l|lll|lll|r|}
						\multicolumn{1}{c}{}&	\multicolumn{3}{c}{vertices} & \multicolumn{3}{c}{edges}& \multicolumn{1}{c}{}\\
						\cline{2-7}
						\multicolumn{1}{l|}{case}    & 2 & - & - & - & - & - & \multicolumn{1}{c}{cost} \\
						\cline{1-8}
						1 & 2  & 0 & 0 & 0 & 0 & 0 & 1   \\
						2 & 5 & 0 & 0 & 0 & 0 & 0 & 0   \\
						\cline{1-8}
					\end{tabular}
		\end{table}
	\subsection{ Construction of the status table for an \textsc{introduce} bag}
		An \textsc{introduce} bag $X_i$  has one more vertex than its child bag $X_{i-1}$ as well as the edges that are adjacent to the new vertex and the vertices of the previous bag.  Assume  $ Stable_{i-1}$ contains all of the possible states that can occur up to this level for visited edges and vertices.
		After adding the new vertex and the corresponding edges we need to add all of the new possible states to the
		new $ Stable_i$. To do so, we compute  $Stable_{i} $ as $ Stable_{i-1} \otimes Btable_{i}$.  Assume that $n$
		is the size of bag $X_{i}$, $e$ is the number of edges with both ends in $X_{i}$,
		$r_{Stable_{i-1}}(\ell_{1},:)$  and $r_{Btable_{i}}(\ell_{2},:)$ are two rows of $ Stable_{i-1}$ and
		$Btable_{i}$, and $j$ refers to the entries of a row of a table.  Then,  Equation~\ref{seq1} describes the construction of the entries of $\ell$th row of  $Stable_{i}$.
		\begin{equation} \label{seq1}
		r_{Stable_{i}}(\ell,j) =\begin{cases}
		r_{Stable_{i-1}}(\ell_{1},j)  \star_{Int} r_{Btable_{i}}(\ell_{2},j), & \text{if $0 \leq j \leq \tw $ },\\
		r_{Stable_{i-1}}(\ell_{1},j)  \ast_{Int} r_{Btable_{i}}(\ell_{2},j), & \text{if $\tw+1 \leq j \leq
			\tw+\binom{\tw+1}{2} $ },\\
		r_{Stable_{i-1}}(\ell_{1},j)  $ + $ r_{Btable_{i}}(\ell_{2},j) $ - $
		|A| $ - $ |B|, &\text{if $ j = \tw+\binom{\tw+1}{2}+1$ }, \\ 
		\end{cases}
		\end{equation}
		where
		\begin{equation}
		A=\{\alpha \mid (r_{Stable_{i-1}}(\ell_1,\alpha) \leq   2   \wedge  \hspace{3pt}
		r_{Btable_{i}}(\ell_2,\alpha)\leq  2)   \wedge  (0 \leq \alpha \leq \tw )  \},
		\end{equation}
		\begin{equation}
		B = \left\{\alpha \mid \left(r_{Stable_{i-1}}(\ell_1,\alpha) =   1   \wedge  \hspace{3pt}
		r_{Btable_{i}}(\ell_2,\alpha) 
		=  1)   \wedge ( \tw+1 \leq \alpha \leq \tw  + {\textstyle \binom{\tw+1}{2}} \right) \right\}.
		\end{equation}
		In Equation~\ref{seq1}, two multiplication operators  $\star_{Int}$ and $\ast_{Int}$ are used to compute the
		entries of $Stable_{i}$. The multiplication tables for these operators are given in  Tables~\ref{starint}
		and~\ref{astint}.  Note that value $ ``-" $   \, in Tables~\ref{starint},~\ref{astint},~\ref{Joinstar}
		and~\ref{Joinast} never happen.  The Algorithm~\ref{Introducealgorithm} describes aforementioned approach in constructing rows of  $ Stable_{i}$ in a formal manner. 
		Algorithms~\ref{Introducealgorithm},~\ref{forgetal} and~\ref{joinalgorithm}  calculate the value of $\ell$
		using Equation~\ref{ell}.  After constructing an entire row in $ Stable_{i}$, the value of $\ell$ is obtained
		as follows: Assume that $(r_{0}, r_{1}, ..., r_{\tw}, r_{\tw+1}, ..., r_{\tw+\binom{\tw+1}{2}},
		r_{\tw+\binom{\tw+1}{2}+1})$ is the output of the algorithms, then number $\ell$ shows  the number of a row
		that this output is saved in it. Equation~\ref{ell} shows how the value of $\ell$ is calculated.  Array $r_{HP}$ is used to save entries of a row until the value of $\ell$ is computed.
		If the $\ell$th row in $ Stable_{i}$ is empty, the algorithm saves the output. However if the $\ell$th row
		was filled, and if the new $r_{\tw+\binom{\tw+1}{2}+1}$ is less than the existing one, the algorithm replaces
		the value of $r_{\tw+\binom{\tw+1}{2}+1}$.  
		\begin{equation}\label{ell}
		\ell=\left( \sum_{i=0}^{\tw} r_i\times9^{i}\right)  + \left(\left(\sum_{j=1}^{\binom{\tw+1}{2}}
		r_{\tw+j}\times4^{j-1}\right) \times 9^{\tw+1}\right).
		\end{equation}
		
				\begin{table}
					\label{starint}
					\caption{Multiplication operation for vertices in \textsc{introduce} bag.}
					\centering
					\begin{tabular}{ r|c|c|c|c|c|c|c|c|c|c| } 
											\multicolumn{1}{r}{$\star_{Int}$}
											&  \multicolumn{1}{c}{0}
											&  \multicolumn{1}{c}{1}
											&  \multicolumn{1}{c}{2}
											&  \multicolumn{1}{c}{3}
											&  \multicolumn{1}{c}{4} 
											&  \multicolumn{1}{c}{5}
											&  \multicolumn{1}{c}{6}
											&  \multicolumn{1}{c}{7}
											&  \multicolumn{1}{c}{8}
											&  \multicolumn{1}{c}{9} \\
											\cline{2-11}
											0 & 0 & - & - & - & - & - & - & - & - & - \\
											\cline{2-11}
											1 & 1 & 1 & 1 & 1 & 1 & 1 & 1 & 1 & 1 & 1 \\
											\cline{2-11}
											2 & 2 & 1 & 2 & 1 & 2 & 2 & 2 & 2 & 2 & 2 \\
											\cline{2-11}
											3 & 3 & 1 & 1 & 3 & 3 & 3 & 3 & 3 & 3 & 3 \\
											\cline{2-11}
											4 & 4 & 1 & 2 & 3 & 4 & 4 & 4 & 4 & 8 & 8 \\
											\cline{2-11}
											5 & 4,5 & 1 & 2 & 3 & 4 & 5 & 4 & 5 & 8 & 9 \\
											\cline{2-11}
											6 & 4,6 & 1 & 2 & 3 & 4,6 & 4,6 & 4,6 & 4,6 & 8 & 8 \\
											\cline{2-11}
											7 & 4,5,6,7 & 1 & 2 & 3 & 4,6 & 5,7 & 4,6 & 5,7 & 8 & 9 \\
											\cline{2-11}
						\end{tabular}
				\end{table}
				\begin{table}
					\label{astint}
					\caption{Multiplication operation for edges in \textsc{introduce} bag.}
					\centering
					\begin{tabular}{ r|c|c|c|c|}
											\multicolumn{1}{r}{$\ast_{Int}$}
											&  \multicolumn{1}{c}{0}
											&  \multicolumn{1}{c}{1}
											&  \multicolumn{1}{c}{2}
											&  \multicolumn{1}{c}{3}
											\\
											\cline{2-5}
											0 & 0 & - & - & -   \\
											\cline{2-5}
											1 & 1 & 1 & 1 & 1  \\
											\cline{2-5}
											2 & 2 & 1 & 2 & 2  \\
											\cline{2-5}
											3 & 2,3 & 1 & 2 & 3 \\
											\cline{2-5}
					\end{tabular}
				\end{table}

		The Algorithm~\ref{Introducealgorithm} takes a row from $Stable_{i-1}$ and a row from $Btable_{i}$, and uses
		Equation~\ref{starint} to fill elements $r_0, ..., r_{\tw}$ of $Stable_{i}$. 
		Given that some of the cells in Table~\ref{starint}  have two values, the exact amount is determined according 
		to the value of an adjacent edge.  Algorithm~\ref{Introducealgorithm} at first determines elements $r_0, ...,
		r_{\tw}$ of $Stable_{i}$ if the cell in Table~\ref{Joinstar}  has one value.
		Similarly, it uses Equation~\ref{astint} to fill elements $r_{\tw+1}, ..., r_{\tw+\binom{\tw+1}{2}}$ of
		$Stable_{i}$. Then it assigns $r_0, ..., r_{\tw}$ which the cell in Table~\ref{starint} has two values, 
		it chooses one of them according to the value of the adjacent edges of a vertex in bag $X_i$. 
		The value of the $r_{\tw+\binom{\tw+1}{2}+1}$ is computed during the assignment of the elements $r_{0}, ...,
		r_{\tw+\binom{\tw+1}{2}}$.
		Computing  $Stable_{2} $ as $ Stable_{1} \otimes Btable_{2}$ is shown in Table~\ref{intoducestable1}.
		Algorithm~\ref{Introducealgorithm}  computes $r_{Stable_{1}}(4,:) \otimes r_{Btable_{2}}(1,:)$ as follows:
		\begin{itemize}
			\item[1) ] $(\underline{2},0,0,0,0,0,1)  \star_{Int} (\underline{3},3,0,1,0,0,1) \rightarrow (\underline{1},0,0,0,0,0,0)$
			\item[2) ] $(2,\underline{0},0,0,0,0,1)  \star_{Int} (3,\underline{3},0,1,0,0,1) \rightarrow (1,\underline{3},0,0,0,0,0)$
			\item[3) ] $(2,0,0,\underline{0},0,0,1)  \ast_{Int} (3,3,0,\underline{1},0,0,1) \rightarrow (1,3,0,\underline{1},0,0,0)$
			\item[4) ] $(2,0,0,0,0,0,\underline{1}) \times   (3,3,0,1,0,0,\underline{1}) \rightarrow (1,3,0,1,0,0,\underline{2})$
		\end{itemize} 
		Also it computes $r_{Stable_{1}}(8,:) \otimes r_{Btable_{2}}(2,:)$ as follows:
		\begin{itemize}
			\item[1) ] $(\underline{5},0,0,0,0,0,0)  \star_{Int} (\underline{7},7,0,3,0,0,0) \rightarrow (\underline{?},0,0,0,0,0,0)$
			\item[2) ] $(5,\underline{0},0,0,0,0,0)  \star_{Int} (7,\underline{7},0,3,0,0,0) \rightarrow (?,\underline{?},0,0,0,0,0)$
			\item[3) ] $(5,0,0,\underline{0},0,0,0)  \ast_{Int} (7,7,0,\underline{3},0,0,0) \rightarrow (?,?,0,\underline{3},0,0,0)$
			\item[4) ] $(\underline{5},0,0,0,0,0,0)  \star_{Int} (\underline{7},7,0,3,0,0,0) \rightarrow (\underline{7},0,0,3,0,0,0)$
			\item[5) ] $(5,\underline{0},0,0,0,0,0)  \star_{Int} (7,\underline{7},0,3,0,0,0) \rightarrow (7,\underline{7},0,3,0,0,0)$
			\item[6) ] $(5,0,0,0,0,0,\underline{0})  \times (7,7,0,3,0,0,\underline{0}) \rightarrow (7,7,0,,0,0,\underline{0})$
		\end{itemize} 
		\begin{table}[pbht]
			\label{leaftable1}
			\caption{$Stable_{1} $}
			\centering
			\begin{tabular}{|l|lll|lll|r|}
				\multicolumn{1}{c}{}&	\multicolumn{3}{c}{vertices} & \multicolumn{3}{c}{edges}& \multicolumn{1}{c}{}\\
				\cline{2-7}
				\multicolumn{1}{l|}{case}    & 2 & - & - & - & - & - & \multicolumn{1}{c}{cost} \\
				\cline{1-8}
				1 & 2  & 0 & 0 & 0 & 0 & 0 & 1   \\
				2 & 5 & 0 & 0 & 0 & 0 & 0 & 0   \\
				\cline{1-8}
			\end{tabular}
		\end{table}
		\begin{table}[pbht]
			\label{intoducebtable1}
			\caption{$Btable_{2}$}
			\centering
			\begin{tabular}{|l|lll|lll|r|}
									\multicolumn{1}{c}{}&	\multicolumn{3}{c}{vertices} & \multicolumn{3}{c}{edges}& \multicolumn{1}{c}{}\\
									\cline{2-7}
									\multicolumn{1}{l|}{case}    & 2 & 3 & - & 4 & - & - & \multicolumn{1}{c}{cost} \\
									\cline{1-8}
									1 & 1 & 1 & 0 & 1 & 0 & 0  & 3 \\
									2 & 1 & 3 & 0 & 1 & 0 & 0  & 2 \\
									3 & 3 & 1 & 0 & 1 & 0 & 0  & 2 \\
									4 & 3 & 3 & 0 & 1 & 0 & 0   & 1 \\
									5 & 2 & 2 & 0 & 2 & 0 & 0   & 2 \\
									6 & 2 & 4 & 0 & 2 & 0 & 0  & 1\\
									7 & 4 & 2 & 0 & 2 & 0 & 0  & 1\\
									8 & 7 & 7 & 0 & 3 & 0 & 0  & 0 \\
									\cline{1-8}
								\end{tabular}
		\end{table}
		\begin{table}[pbht]
			\centering
			\label{intoducestable1}	
			\caption{$Stable_{2} $ }
			\begin{tabular}{|l|lll|lll|r|}
									\multicolumn{1}{c}{}&	\multicolumn{3}{c}{vertices} & \multicolumn{3}{c}{edges}& \multicolumn{1}{c}{}\\
									\cline{2-7}
									\multicolumn{1}{c|}{case}    & 2 & 3 & - & 4 & - & - & \multicolumn{1}{c}{cost} \\
									\cline{1-8}
									1 & 1 & 1 & 0 & 1 & 0 & 0  & 3 \\
									2 & 1 & 3 & 0 & 1 & 0 & 0  & 2 \\
									3 & 3 & 1 & 0 & 1 & 0 & 0  & 2 \\
									4 & 3 & 3 & 0 & 1 & 0 & 0   & 1 \\
									5 & 2 & 2 & 0 & 2 & 0 & 0   & 2 \\
									6 & 2 & 4 & 0 & 2 & 0 & 0  & 1\\
									7 & 4 & 2 & 0 & 2 & 0 & 0  & 1\\
									8 & 7 & 7 & 0 & 3 & 0 & 0  & 0 \\
									
									\cline{1-8}
								\end{tabular}	
		\end{table}
		\begin{algorithm}	
			\caption{ Status table's row construction algorithm for an \textsc{introduce} bag.}
			\begin{algorithmic}
				\\ INPUT: $r_{Stable_{i-1}}(\ell_1,:)$ and  $r_{Btable_{i}}(\ell_2,:)$
				\\ OUTPUT: $r_{Stable_{i}}(\ell,:)$
				\State $c\gets0$
				\For{$ j \gets 0  $ to $ \tw$}
				\If{$(r_{Stable_{i-1}}(\ell_1,j)  \leq 2) \textbf{ and } (r_{Btable_{i}}(\ell_2,j) \leq 2)$} 
				\State $c\gets c+1$
				\EndIf
				\If{$|r_{Stable_{i-1}}(\ell_1,j)  \star_{Int} r_{Btable_{i}}(\ell_2,j)|=1$} 
				\State $r_{HS}(1,j) \gets r_{Stable_{i-1}}(\ell_1,j)  \star_{int} r_{Btable_{i}}(\ell_2,j)$
				\EndIf  
				\EndFor
				\For{$ j \gets \tw+1$ to $\tw+\binom{\tw+1}{2}$}
				\If{$(r_{Stable_{i-1}}(\ell_1,j) =  1) \textbf{ and } (r_{Btable_{i}}(\ell_2,j) = 1)$} 
				\State $c\gets c+1$
				\EndIf
				\If{$|r_{Stable_{i-1}}(\ell_1,j)  \ast_{Int} r_{Btable_{i}}(\ell_2,j)|=1$}
				\State  $r_{HS}(1,j) \gets r_{Stable_{i-1}}(\ell_1,j)  \ast_{Int} r_{Btable_{i}}(\ell_2,j)$
				\Else 
				\If{ for $e_{j}=v_{r}v_{new}, Stable_{i-1}(\ell_1,r) \leq 3 $}
				\State $r_{HS}(1,j) \gets 2$
				\Else \State $r_{HS}(1,j) \gets 3$
				\EndIf 
				\EndIf
				\EndFor
				\For{$ j \gets 0$ to $\tw$}
				\If{$|r_{Stable_{i-1}}(\ell_1,j)  \star_{Int} r_{Btable_{i}}(\ell_2,j)|=2$}
				\If{$r_{Stable_{i-1}}(\ell_1,j) = 0 \textbf{ and } r_{Btable_{i}}(\ell_2,j)=5 $} 
				\If{$\exists_{w_{r} \in N^{md}_{Bag}(v_{j})}, r_{HS}(1,r) \leq 2 $}\State $r_{HS}(1,j) \gets 4$
				\Else \State
				$r_{HS}(1,j) \gets 5$
				\EndIf 
				\Else 
				\If{$\forall_{w_{r} \in N^{md}_{Bag}(v_{j}), e_{s}=(w_{r}v_{j})},  r_{Stable_{i}}(s) \leq 2$}\State $r_{HS}(1,j) \gets \text{first element of }r_{Stable_{i-1}}(\ell_1,j)  \ast_{Int} r_{Btable_{i}}(\ell_2,j)$
				\Else \State 
				$r_{HS}(1,j) \gets\text{second element of } r_{Stable_{i-1}}(\ell_1,j)  \ast_{Int} r_{Btable_{i}}(\ell_2,j)$
				\EndIf 
				\EndIf 
				\Else 
				
				\If{$ \forall_{w_{r} \in N^{md}_{Bag}(v_{j}), e_{s}=(w_{r}v_{j})} ,r_{HS}(1,s) \leq 2  \text{\bf{ and }} \exists_{w_{r} \in N^{md}_{Bag}(v_{j})} ,r_{HS}(1,r) \leq 2$} 
				\State $r_{HS}(1,j) \gets 4$					
				\ElsIf{$ \forall_{w_{r} \in N^{md}_{Bag}(v_{j}), e_{s}=(w_{r}v_{j})} ,r_{HS}(1,s) \leq 2  \text{\bf{ and }} \forall_{w_{r} \in N^{md}_{Bag}(v_{j})} ,r_{HS}(1,r) \geq 3$}
				\State $r_{HS}(1,j) \gets 5$					
				\algstore{myalg}
			\end{algorithmic}\label{Introducealgorithm} 
		\end{algorithm}
		\begin{algorithm}                     
			\begin{algorithmic} [1]      
				\algrestore{myalg}
				
				\ElsIf{$ \exists_{w_{r} \in N^{md}_{Bag}(v_{j}), e_{s}=(w_{r}v_{j})} ,r_{HS}(1,s) = 3  \text{\bf{ and }} \exists_{w_{r} \in N^{md}_{Bag}(v_{j})} ,r_{HS}(1,r) \leq 2$} 
				\State $r_{HS}(1,j) \gets 6$					
				\ElsIf{$ \exists_{w_{r} \in N^{md}_{Bag}(v_{j}), e_{s}=(w_{r}v_{j})} ,r_{HS}(1,s) = 3  \text{\bf{ and }} \forall_{w_{r} \in N^{md}_{Bag}(v_{j})} ,r_{HS}(1,r) \geq 3$} 
				\State $r_{HS}(1,j) \gets 7$					
				\EndIf 
				\EndIf 
				\EndFor
				\State Calculate the value of $\ell$
				\State $r_{Stable_{i}}(\ell,1:end-1) \gets r_{HS}(1,1:end-1)$
				\If{$r_{Stable_{i}}(\ell,end) > (r_{Stable_{i-1}}(\ell_1,end) + r_{Btable_{i}}(\ell_2,end) - c)$}
				\State $r_{Stable_{i}}(\ell,end) \gets r_{Stable_{i-1}}(\ell_1,end) + r_{Btable_{i}}(\ell_2,end) - c$
				\EndIf
			\end{algorithmic}
		\end{algorithm}
		
		\subsection{Construction of the status table for a \textsc{forget} bag}
		A \textsc{forget} bag $X_i$ loses one vertex and its incident edges with respect to its present bag
		$X_{i-1}$. So, it is enough to omit the invalid rows from $ Stable_{i-1}$ to obtain the  $ Stable_{i}$ for
		bag  $X_{i}$. Algorithm~\ref{forgetal} describes an approach for constructing rows of $ Stable_{i}$. This
		algorithm  takes a row from $Stable_{i-1}$. Let vertex $v_{eliminated}$ is the vertex that will be deleted in
		bag $X_i$. If this vertex has values 5, 7, 8 and 9, the algorithm omits this row, but if it has value 6, the
		algorithm updates the value of $N(v)$ in bag $X_{i-1}$. When its neighbor is covered but the edge between
		them is not covered, the neighbor's value changes to 8, and  when none of its neighbors and 
		edges between them are not covered, the neighbor's value changes to 9. 
		For example Algorithm~\ref{forgetal} takes rows of $ Stable_{11}$ and constructs $ Stable_{12}$.
		
		\begin{table}[pbht]
					\centering
			\label{jointable}	
			\caption{$Stable_{11}$}
			\scalebox{0.7}{  
										\begin{tabular}{|l|lll|lll|r|}
											\multicolumn{1}{c}{}&	\multicolumn{3}{c}{vertices} & \multicolumn{3}{c}{edges}& \multicolumn{1}{c}{}\\
											\cline{2-7}
											\multicolumn{1}{c|}{case}    & 1 & 4 & - & 3 & - & - & \multicolumn{1}{c}{cost} \\
											\cline{1-8}
											1 & 1 & 1 & 0 & 1 & 0 & 0  & 4 \\
											2 & 1 & 2 & 0 & 2 & 0 & 0  & 3\\
											3 & 1 & 3 & 0 & 1 & 0 & 0  & 4\\
											4 & 1 & 3 & 0 & 2 & 0 & 0  & 3 \\
											5 & 1 & 4 & 0 & 2 & 0 & 0  & 3\\
											6 & 1 & 8 & 0 & 2 & 0 & 0  & 3\\
											7 & 2 & 1 & 0 & 2 & 0 & 0  & 3\\
											8 & 2 & 2 & 0 & 2 & 0 & 0  & 3\\
											9 & 2 & 3 & 0 & 2 & 0 & 0  & 3\\
											10 & 2 & 4 & 0 & 2 & 0 & 0  & 3\\
											11 & 2 & 8 & 0 & 2 & 0 & 0  & 3\\
											12 & 3 & 1 & 0 & 1 & 0 & 0  & 3\\
											13 & 3 & 1 & 0 & 2 & 0 & 0  & 3\\
											14 & 3 & 2 & 0 & 2 & 0 & 0  & 2\\
											15 & 3 & 3 & 0 & 1 & 0 & 0  & 3\\
											16 & 3 & 3 & 0 & 2 & 0 & 0  & 3\\
											17 & 3 & 4 & 0 & 2 & 0 & 0  & 3\\
											18  & 3 & 8 & 0 & 2 & 0 & 0  & 3\\
											19 & 4 & 1 & 0 & 2 & 0 & 0  & 3\\
											20 & 4 & 2 & 0 & 2 & 0 & 0  & 2\\
											21 & 4 & 3 & 0 & 2 & 0 & 0  & 3\\
											22 & 5 & 3 & 0 & 2 & 0 & 0  & 2\\ 
											23 & 6 & 6 & 0 & 3 & 0 & 0  & 3 \\
											24 & 7 & 6 & 0 & 3 & 0 & 0  & 2\\
											25 & 7 & 7 & 0 & 3 & 0 & 0  & 1\\
											26 & 8 & 1 & 0 & 2 & 0 & 0  & 3\\
											27 & 8 & 2 & 0 & 2 & 0 & 0  & 2\\
											28 & 8 & 3 & 0 & 2 & 0 & 0  & 2\\
											29 & 8 & 6 & 0 & 3 & 0 & 0  & 2\\
											30 & 8 & 8 & 0 & 2 & 0 & 0  & 2\\
											\cline{1-8}
										\end{tabular} 
									}	
		\end{table}
		\begin{table}[pbht]
					\centering
				\label{roottable}
			\caption{$Stable_{12} $}
			\scalebox{.7}{  
										\begin{tabular}{|l|lll|lll|r|}
											\multicolumn{1}{c}{}&	\multicolumn{3}{c}{vertices} & \multicolumn{3}{c}{edges}& \multicolumn{1}{c}{}\\
											\cline{2-7}
											\multicolumn{1}{c|}{case}    & 1 & - & - & - & - & - & \multicolumn{1}{c}{cost} \\
											\cline{1-8}
											1 & 1  & 0 & 0 & 0 & 0 & 0 & 3  \\
											2 & 2 & 0 & 0 & 0 & 0 & 0 & 3   \\
											1 & 3  & 0 & 0 & 0 & 0 & 0 & 2  \\
											2 & 4 & 0 & 0 & 0 & 0 & 0 & 2  \\
											1 & 5  & 0 & 0 & 0 & 0 & 0 & 2  \\
											2 & 8 & 0 & 0 & 0 & 0 & 0 & 2  \\
											\cline{1-8}
										\end{tabular}
									}
		\end{table}
		
		\begin{algorithm}	[h!]
					\caption{ Status table's row construction algorithm for a \textsc{forget} bag. }
					\begin{algorithmic}
						\\ INPUT: $r_{Stable_{i-1}}(\ell,:)$ 
						\\ OUTPUT: $r_{Stable_{i}}(\ell,:)$ or 0
						\If{$(r_{Stable_{i-1}}(\ell,v_{eliminated})  \geq7) $ \bf{or} $  (r_{Stable_{i-1}}(\ell,v_{eliminated})=  5)$}
						\State  Return 0  
						\Else
						\State $r_{HS}(1,:) \gets r_{Stable_{i-1}}(\ell,:)$
						\If {$r_{Stable_{i-1}}(\ell,v_{eliminated})  = 6$}
						\If {$( \forall_{v_{r} \in N^{md}_{X_{i-1}}(v_{eliminated}),e_{s}=(v_{r},v_{eliminated})},r_{Stable_{i-1}}(\ell,r)=6) \text{\bf{ and }} (r_{Stable_{i-1}}(\ell,s) = 3)$}
						\State $r_{HS}(\ell,r)\gets8$
						\ElsIf{$(\forall_{v_{r} \in N^{md}_{Bag_{i-1}}(v_{eliminated}),e_{s}=(v_{r},v_{eliminated})},r_{Stable_{i-1}}(\ell,r)=7)\text{\bf{ and }}(r_{Stable_{i-1}}(\ell,s)=3)$}
						\State $r_{HS}(\ell,r)\gets9$
						\EndIf
						\EndIf
						
						\EndIf 
						\State Calculate the value of $\ell$
						\State $r_{Stable_{i}}(\ell,:) \gets r_{HS}(1,:)$
						\State $r_{HS}(\ell,v_{eliminated}) \gets 0$
						\State $ \forall_{v_{r} \in N^{md}_{X_{i-1}}(v_{eliminated}), e_{s}=(v_{r},v_{eliminated})}, r_{HS}(\ell,s) \gets 0 $
						
					\end{algorithmic} \label{forgetal}
				\end{algorithm} 
				\subsection{Construction of the status table for a \textsc{join} bag}
				
				A \textsc{join} bag $X_i$ has the same set of vertices and edges with  its two children ${X_{i1}}$ and ${X_{i2}}$.  To construct possible states for $X_i$ in  $ Stable_{i}$, we compute $Stable_{i} = Stable_{i1} \otimes Stable_{i2}$. Let 
						$r_{Stable_{i1}}(\ell_{1},:)$  and $r_{Stable_{i2}}(\ell_{2},:)$ be two rows of $ Stable_{i1}$ and
						$Stable_{i2}$, and $j$ refers to the entries of a row of a table.  Equation~\ref{seq1} describes how to construct the rows of  $ Stable_{i}$.
						\begin{equation} \label{seq2}
								r_{Stable_{i}}(\ell,j) =\begin{cases}
								r_{Stable_{i1}}(\ell_1,j)  \star_{Join} r_{Stable_{i2}}(\ell_2,j), & \text{if $0 \leq j \leq \tw$ }.\\
								r_{Stable_{i1}}(\ell_1,j)  \ast_{Join} r_{Stable_{i2}}(\ell_2,j), & \text{if $\tw+1 \leq i \leq \tw+\binom{\tw+1}{2}$ }.\\
								r_{Stable_{i1}}(\ell_1,j)  $ + $  r_{Stable_{i2}}(\ell_2,j)  $ - $
								|A| $ - $ |B|, &\text{if $ i = \tw+\binom{\tw+1}{2}+1$ }. \\ 
								\end{cases}
								\end{equation}
								where
								\begin{equation}
								A= \left\{\alpha\mid (r_{Stable_{i1}}(\ell_1,\alpha) \leq   2   \wedge  \hspace{3pt}
								r_{Stable_{i2}}(\ell_2,\alpha)\leq  2)   \wedge  (0 \leq \alpha \leq \tw)  \right\}.
								\end{equation}
								\begin{equation}
								B =  \left\{\alpha \mid \left(r_{Stable_{i1}}(\ell_1,\alpha) =   1   \wedge  \hspace{3pt}
								r_{Stable_{i2}}(\ell_2,\alpha) =  1)   \wedge  (\tw+1 \leq \alpha \leq \tw +1+ {\textstyle \binom{\tw+1}{2}} \right) \right\}.
								\end{equation}
				In Equation~\ref{seq2},  two different multiplication operations  $\star_{Join}$ and $\ast_{Join}$ are
						used to obtain the entries of $ Stable_{i}$. Their multiplication tables  are given in
						Tables~\ref{Joinstar} and~\ref{Joinast}. Algorithm~\ref{joinalgorithm} describes how to construct the rows of  $ Stable_{i}$ precisely.
						Given that some of the cells in Table~\ref{Joinstar}  have two values we do as before.
						Algorithm~\ref{joinalgorithm} at first determines elements $r_0, ..., r_{\tw}$ of $Stable_{i}$ if the cell in
						Table~\ref{Joinstar}  has one value then it uses Table~\ref{Joinast} to fill elements $r_{\tw+1}, ...,
						r_{\tw+\binom{\tw+1}{2}}$ of $Stable_{i}$. Finally, it assigns $r_0, ..., r_{\tw}$ when the cell in 
						Table~\ref{Joinstar} has two values, it chooses one as mentioned before.  The value of $r_{\tw+\binom{\tw+1}{2}+1}$
						that is computed during the assignment of the elements is $r_{0}, ..., r_{\tw+\binom{\tw+1}{2}}$.
				\begin{table}[htpb]
				\centering
					\label{Joinstar}
					\caption{Multiplication operation for vertices in \textsc{join} bag.}
					\scalebox{0.7}{  
												\begin{tabular}{ r|c|c|c|c|c|c|c|c|c|c| }
													\multicolumn{1}{r}{$\star_{Join}$}
													&  \multicolumn{1}{c}{0}
													&  \multicolumn{1}{c}{1}
													&  \multicolumn{1}{c}{2}
													&  \multicolumn{1}{c}{3}
													&  \multicolumn{1}{c}{4} 
													&  \multicolumn{1}{c}{5}
													&  \multicolumn{1}{c}{6}
													&  \multicolumn{1}{c}{7}
													&  \multicolumn{1}{c}{8}
													&  \multicolumn{1}{c}{9} \\
													\cline{2-11}
													0 & 0 & - & - & - & - & - & - & - & - & - \\
													\cline{2-11}
													1 & - & 1 & 1 & 1 & 1 & 1 & 1 & 1 & 1 & 1 \\
													\cline{2-11}
													2 & - & 1 & 2 & 1 & 2 & 2 & 2 & 2 & 2 & 2 \\
													\cline{2-11}
													3 & - & 1 & 1 & 3 & 3 & 3 & 3 & 3 & 3 & 3 \\
													\cline{2-11}
													4 & - & 1 & 2 & 3 & 4 & 4 & 4 & 4 & 8 & 8 \\
													\cline{2-11}
													5 & - & 1 & 2 & 3 & 4 & 5 & 4 & 5 & 8 & 9 \\
													\cline{2-11}
													6 & - & 1 & 2 & 3 & 4 & 4 & \begin{tabular}{l|*{2}{c}}
														\hline
														First & \multicolumn{2}{|c}{Second} \\ \hline
														4 & 6  
													\end{tabular}
													& \begin{tabular}{l|*{2}{c}}
														\hline
														First & \multicolumn{2}{|c}{Second} \\ \hline
														4 & 6  
													\end{tabular} & 8 & 8 \\
													\cline{2-11}
													7 & - & 1 & 2 & 3 & 4 & 5 & 
													\begin{tabular}{l|*{2}{c}}
														\hline
														First & \multicolumn{2}{|c}{Second} \\ \hline
														4 & 6  
													\end{tabular} & \begin{tabular}{l|*{2}{c}}
													\hline
													First & \multicolumn{2}{|c}{Second} \\ \hline
													5 & 7  
												\end{tabular} & 8 & 9 \\
												\cline{2-11}
												8 & - & 1 & 2 & 3 & 8 & 8 & 8 & 8 & 8 & 8 \\
												\cline{2-11}
												9 & - & 1 & 2 & 3 & 8 & 9 & 8 & 9 & 8 & 9 \\
												\cline{2-11}
											\end{tabular}
											
										}	
				\end{table}
				
				\begin{table}[htpb]
				\centering
				\label{Joinast}
					\caption{Multiplication operation for edges in \textsc{join} bag.}
					\scalebox{0.7}{  
											\begin{tabular}{ r|c|c|c|c| c| }
												\multicolumn{1}{r}{$\ast_{Join}$}
												&  \multicolumn{1}{c}{0}
												&  \multicolumn{1}{c}{1}
												&  \multicolumn{1}{c}{2}
												&  \multicolumn{1}{c}{3}
												\\
												\cline{2-5}
												0 & 0 & - & - & -   \\
												\cline{2-5}
												1 & - & 1 & 1 & 1  \\
												\cline{2-5}
												2 & - & 1 & 2 & 2  \\
												\cline{2-5}
												3 & - & 1 & 2 & 3 \\
												\cline{2-5}
											\end{tabular}
											
										}
				\end{table}
				
		Algorithm~\ref{joinalgorithm} computes  $ Stable_{1} \otimes Btable_{2}$ to construct $Stable_{2} $. The status
			table of bag 12 of Figure~\ref{NTD} is shown in Table~\ref{jointable3}.
		\begin{table}
		\label{jointable1}	
		\centering
			\caption{$Stable_{6}$}
			\scalebox{0.7}{  
									\begin{tabular}{|l|lll|lll|r|}
										\multicolumn{1}{c}{}&	\multicolumn{3}{c}{vertices} & \multicolumn{3}{c}{edges}& \multicolumn{1}{c}{}\\
										\cline{2-7}
										\multicolumn{1}{c|}{case}    & 1 & 4 & - & 3 & - & - & \multicolumn{1}{c}{cost} \\
										\cline{1-8}
										1 & 1 & 1 & 0 & 1 & 0 & 0  & 4 \\
										2 & 1 & 2 & 0 & 2 & 0 & 0  & 3\\
										3 & 1 & 3 & 0 & 1 & 0 & 0  & 3\\
										4 & 1 & 3 & 0 & 2 & 0 & 0  & 2 \\
										5 & 1 & 4 & 0 & 2 & 0 & 0  & 2\\
										6 & 1 & 8 & 0 & 2 & 0 & 0  & 2\\
										7 & 2 & 2 & 0 & 2 & 0 & 0  & 3\\
										8 & 2 & 4 & 0 & 2 & 0 & 0  & 2\\
										9 & 2 & 8 & 0 & 2 & 0 & 0  & 2\\
										10 & 3 & 1 & 0 & 1 & 0 & 0  & 3\\
										11 & 3 & 1 & 0 & 2 & 0 & 0  & 3\\
										12 & 3 & 2 & 0 & 2 & 0 & 0  & 2\\
										13 & 3 & 3 & 0 & 1 & 0 & 0  & 2\\
										14 & 3 & 3 & 0 & 2 & 0 & 0  & 2\\
										15 & 3 & 4 & 0 & 2 & 0 & 0  & 2\\
										16 & 3 & 5 & 0 & 2 & 0 & 0  & 2\\
										17  & 3 & 9 & 0 & 2 & 0 & 0  & 2\\
										18 & 4 & 1 & 0 & 2 & 0 & 0  & 2\\
										19 & 4 & 2 & 0 & 2 & 0 & 0  & 2\\
										20 & 4 & 3 & 0 & 2 & 0 & 0  & 2\\
										21 & 5 & 3 & 0 & 2 & 0 & 0  & 2\\ 
										22 & 6 & 6 & 0 & 3 & 0 & 0  & 2 \\
										23 & 6 & 7 & 0 & 3 & 0 & 0  & 2\\
										24 & 7 & 7 & 0 & 3 & 0 & 0  & 1\\
										25 & 8 & 1 & 0 & 2 & 0 & 0  & 3\\
										26 & 8 & 2 & 0 & 2 & 0 & 0  & 2\\
										27 & 8 & 3 & 0 & 2 & 0 & 0  & 2\\
										28 & 8 & 6 & 0 & 3 & 0 & 0  & 1\\
										29 & 8 & 9 & 0 & 3 & 0 & 0  & 1\\
										
										\cline{1-8}
									\end{tabular} 
								} 
		\end{table}
		
		\begin{table}[pbht]
		\centering
		 \label{jointable2}	
		 \caption{$Stable_{10}$}
			\scalebox{0.7}{  
									\begin{tabular}{|l|lll|lll|r|}
										\multicolumn{1}{c}{}&	\multicolumn{3}{c}{vertices} & \multicolumn{3}{c}{edges}& \multicolumn{1}{c}{}\\
										\cline{2-7}
										\multicolumn{1}{c|}{case}    & 1 & 4 & - & 3 & - & - & \multicolumn{1}{c}{cost} \\
										\cline{1-8}
										1 & 1 & 1 & 0 & 1 & 0 & 0  & 3 \\
										2 & 1 & 3 & 0 & 1 & 0 & 0  & 3\\
										3 & 2 & 2 & 0 & 2 & 0 & 0  & 2\\
										4 & 2 & 3 & 0 & 2 & 0 & 0  & 2\\
										5 & 2 & 4 & 0 & 2 & 0 & 0  & 2\\
										6 & 2 & 3 & 0 & 2 & 0 & 0  & 2\\
										
										7 & 3 & 1 & 0 & 1 & 0 & 0  & 2\\
										8 & 3 & 3 & 0 & 1 & 0 & 0  & 2\\ 
										9 & 4 & 1 & 0 & 2 & 0 & 0  & 2\\
										10 & 4 & 2 & 0 & 2 & 0 & 0  & 1\\
										11 & 5 & 3 & 0 & 2 & 0 & 0  & 1\\ 
										12 & 7 & 6 & 0 & 3 & 0 & 0  & 1\\
										\cline{1-8}
									\end{tabular}
								} 
		\end{table}
		
		\begin{table}[hbpt]
			 \label{jointable3}	
			 \centering
			\caption{$Stable_{11}$}
			\scalebox{0.7}{  
									\begin{tabular}{|l|lll|lll|r|}
										\multicolumn{1}{c}{}&	\multicolumn{3}{c}{vertices} & \multicolumn{3}{c}{edges}& \multicolumn{1}{c}{}\\
										\cline{2-7}
										\multicolumn{1}{c|}{case}    & 1 & 4 & - & 3 & - & - & \multicolumn{1}{c}{cost} \\
										\cline{1-8}
										1 & 1 & 1 & 0 & 1 & 0 & 0  & 4 \\
										2 & 1 & 2 & 0 & 2 & 0 & 0  & 3\\
										3 & 1 & 3 & 0 & 1 & 0 & 0  & 4\\
										4 & 1 & 3 & 0 & 2 & 0 & 0  & 3 \\
										5 & 1 & 4 & 0 & 2 & 0 & 0  & 3\\
										6 & 1 & 8 & 0 & 2 & 0 & 0  & 3\\
										7 & 2 & 1 & 0 & 2 & 0 & 0  & 3\\
										8 & 2 & 2 & 0 & 2 & 0 & 0  & 3\\
										9 & 2 & 3 & 0 & 2 & 0 & 0  & 3\\
										10 & 2 & 4 & 0 & 2 & 0 & 0  & 3\\
										11 & 2 & 8 & 0 & 2 & 0 & 0  & 3\\
										12 & 3 & 1 & 0 & 1 & 0 & 0  & 3\\
										13 & 3 & 1 & 0 & 2 & 0 & 0  & 3\\
										14 & 3 & 2 & 0 & 2 & 0 & 0  & 2\\
										15 & 3 & 3 & 0 & 1 & 0 & 0  & 3\\
										16 & 3 & 3 & 0 & 2 & 0 & 0  & 3\\
										17 & 3 & 4 & 0 & 2 & 0 & 0  & 3\\
										18  & 3 & 8 & 0 & 2 & 0 & 0  & 3\\
										19 & 4 & 1 & 0 & 2 & 0 & 0  & 3\\
										20 & 4 & 2 & 0 & 2 & 0 & 0  & 2\\
										21 & 4 & 3 & 0 & 2 & 0 & 0  & 3\\
										22 & 5 & 3 & 0 & 2 & 0 & 0  & 2\\ 
										23 & 6 & 6 & 0 & 3 & 0 & 0  & 3 \\
										24 & 7 & 6 & 0 & 3 & 0 & 0  & 2\\
										25 & 7 & 7 & 0 & 3 & 0 & 0  & 1\\
										26 & 8 & 1 & 0 & 2 & 0 & 0  & 3\\
										27 & 8 & 2 & 0 & 2 & 0 & 0  & 2\\
										28 & 8 & 3 & 0 & 2 & 0 & 0  & 2\\
										29 & 8 & 6 & 0 & 3 & 0 & 0  & 2\\
										30 & 8 & 8 & 0 & 2 & 0 & 0  & 2\\
										\cline{1-8}
									\end{tabular}
								}
		\end{table}
		
		\begin{algorithm}
				\caption{Status table's row construction algorithm for a \textsc{join} bag.}
				\begin{algorithmic}
					\\ INPUT: $r_{Stable_{i1}}(\ell_1,:)$ and  $r_{Stable_{i2}}(\ell_2,:)$
					\\ OUTPUT: $r_{Stable_{i}}(\ell,:)$
					\For{$ j \gets 0 $ to $ \tw$}
					\If{$|r_{Stable_{i1}}(\ell_1,j)  \star_{Join} r_{Stable_{i2}}(\ell_2,j)|=1$}
					\State  $r_{HS}(1,j) \gets r_{Stable_{i1}}(\ell_1,j)  \star_{Join} r_{Stable_{i2}}(\ell_2,j)$
					\EndIf  
					\If{$(r_{Stable_{i1}}(\ell_1,j)  \leq 2) \textbf{ and } (r_{Stable_{i}}(\ell_2,j) \leq 2)$} 
					\State $c\gets c+1$
					\EndIf
					
					\EndFor
					\For{$ j \gets \tw$ to $\tw+\binom{\tw+1}{2}$}
					\State			 $r_{HS}(1,j) \gets r_{Stable_{i1}}(\ell_1,j)  \ast_{Join} r_{Stable_{i2}}(\ell_2,j)$
					\If{$(r_{Stable_{i-1}}(\ell_1,j) =  1) \textbf{ and } (r_{Stable_{i}}(\ell_2,j) = 1)$} 
					\State $c\gets c+1$
					\EndIf
					
					\EndFor
					\For{$ j \gets 0$ to $\tw$}
					\If{$|r_{Stable_{i1}}(\ell_1,j)  \star_{Join} r_{Stable_{i2}}(\ell_2,j)|=2$} 
					
					\If{$ \forall_{w_{r} \in N^{md}_{Bag}(v_{j}), e_{s}=(w_{r}v_{j})}, r_{HS}(1,s) \leq 2$}\State $r_{HS}(\ell,j) \gets \text{First element of }  r_{Stable_{i1}}(\ell_1,j)  \ast_{Join} r_{Stable_{i2}}(\ell_2,j)$
					\Else
					\State$r_{HS}(1,j) \gets \text{Second element of }  r_{Stable_{i1}}(\ell_1,j)  \ast_{Join} r_{Stable_{i2}}(\ell_2,j)$
					
					\EndIf 
					\EndIf 
					
					\EndFor
					\State Calculate the value of $\ell$
					\State	$r_{Stable_{i}}(\ell,1:end-1) \gets r_{HS}(1,1:end-1)$
					\If{$r_{Stable_{i}}(\ell,end) > (r_{Stable_{i-1}}(end) + r_{Btable_{i}}(end) - c)$}
					\State $r_{Stable_{i}}(end) \gets r_{Stable_{i-1}}(end) + r_{Btable_{i}}(end) - c$
					\EndIf
					
				\end{algorithmic} \label{joinalgorithm}
			\end{algorithm}
			
			Note that Algorithm~\ref{MixedAlgorithm} computes the value of  $\gamma_{md}$, however it is not a minimum
				mixed dominating set. It is possible to modify the algorithm to obtain a mixed dominating set with minimum
				size. This modification is as follows.
				We first consider some fixed arbitrary one-to-one total numbering function $\phi$  used to code elements
				of $S \subseteq V \cup E$.
				Let $|V|=n$ and $|E|=m$, $  \phi: \{V \cup E\} \xrightarrow{1-1} \{1, 2, ..., n+m\}$. 
				The function $\phi$ determines an arbitrary order on set $\{V \cup E\}$. We get elements in a particular
				order to code a partial solution.
				Indeed, we display a partial solution as a binary number, i.e.~every element in $\{V \cup E\}$ can have two
				values of 0 or 1 where 0 indicates that the corresponding element is not present (in the mixed domination
				set) and 1 indicates it is present.
				We use a function 
				$\psi : \phi(x) \rightarrow \{0, 1\}$.
				to convert a partial solution $x$ to a binary number. Its enough to change the algorithm and save 
				elements of the partial solution by a binary number just after constructing each state. 
				Other changes to the algorithm are straightforward. 
				Note that the maximum number of bits that can be changed in a state is equal to $k + \binom {k}{2}$, where $k$ is the number of vertices in a bag.
				
				After computing $\gamma_{md}$, our proposed algorithm traverses the tree decomposition via $\tau$ recursively 
				and identifies the edges and vertices in the mixed domination sets. Finally, the AMDS problem is solved.
				
				\subsection{The Modified Algorithm to Solve the MDS Problem}
				In this section, we use the notion of fast subset convolution, which was introduced by Van et al.~in~\cite{van2009dynamic}, to solve MDS in time $O^*(6^{\tw})$. 
					Van et al.~introduced two new techniques, the first using a variant of convolutions, and the
					second being a simple way of partitioned table handling, which can be used for MDS as well.
					They used an alternative representation to obtain an exponentially faster algorithm. 
					In their solution for dominating set problem, each vertex can be in three states as follow:
					\begin{itemize}
						\item[1:] Vertex is in the dominating set.
						\item[$0_1$:] Vertex is not in the dominating set and has already been dominated.
						\item[$0_0$:] Vertex is not in the dominating set and has not yet been dominated.
					\end{itemize}
					
					A vertex in their alternative representation defines two basic states $1$ and $0_?$. The state $0_?$ denotes that a vertex is not in the dominating set and may or may not be dominated by the current dominating set. The transformation can be applied by the formula of $F(c)$, which represents their table for coloring $c$. 
					\[F(c_1\times {0_?} \times c_2) = F(c_1\times {0_0} \times c_2) + F(c_1\times {0_1} \times c_2), \] 
					where $c_1$ is a subcoloring of size $i$, and $c_2$ is a subcoloring of size $k - i - 1$. The transformation to basic states is applicable using
					\[F(c_1\times {0_1} \times c_2) = F(c_1\times {0_?} \times c_2) - F(c_1\times {0_0} \times c_2). \]
					Dynamic programming algorithm over nice tree decomposition uses $O^*(3^\tw)$ time on leaf, \textsc{introduce} and \textsc{forget} bags, however, it uses
					$O^*(4^\tw)$ time on a \textsc{join} bag. By using the results of the convolution of Van et al., 
					the time complexity of a \textsc{join} bag is improved to $O^*(3^\tw)$.
					
					Now, we introduce our algorithm to solve MDS. Similar to solving the AMDS, we assign a power value to the vertices and use dynamic programming. Our solution to the MDS has two differences compared with the solution to the AMDS:
					\begin{itemize}
						\item[1:] The tables $Stable$ and $Btable$ store a valid assignment for just vertices and we do not need to consume memory to store information on edges in a beg.
						\item[2:] A vertex $v \in X_{i} $ has six possible states based on satisfying $v \in \mathcal{MD}$ and being covered by
						$N^{md}_G[v]$. These six conditions are illustrated in Table~\ref{bag condition vertex for msd}. 
					\end{itemize}
					
					\begin{table}[pbht]
					\centering
					\label{bag condition vertex for msd}
						\caption{Possible states for each vertex with respect to a  bag in tables $Stable$ and $Btable$.}
						\scalebox{0.7}{  
									\begin{tabular}{ r|c|c|c|c|c|c|c| }
										\multicolumn{1}{r}{}
										&  \multicolumn{1}{c}{$v\in \mathcal{MD}$}
										&  \multicolumn{1}{c}{$N^{e}_{X_{i}}(v)\in \mathcal{MD}$ }
										
										&  \multicolumn{1}{c}{vertex cover}
										& \multicolumn{1}{c}{edge cover} 
										&  \multicolumn{1}{c}{vertex power}
										
										\\
										\cline{2-6}
										1 & $v \in \mathcal{MD}$ & $\forall e \in N^e_{X_{i}}(v)$  & $v$ is covered & $\forall e \in N^e_{X_{i}}(v), e$ is covered & 2
										\\
										\cline{2-6}
										3 & $v \notin \mathcal{MD}$ & $\exists e \in N^e_{X_{i}}(v) , e \in \mathcal{MD}$ & $v$ is covered & $\forall e \in N^e_{X_{i}}(v), e$ is covered & 1
										\\
										\cline{2-6}
										4 & $v \notin \mathcal{MD}$ & $\forall e \in N^e_{X_{i}}(v), e \notin \mathcal{MD}$ & $v$ is covered & $\forall e \in N^e_{X_{i}}(v), e$ is covered & 0 
										\\
										\cline{2-6}
										5 & $v \notin \mathcal{MD}$ & $\forall e \in N^e_{X_{i}}(v), e \notin \mathcal{MD}$ & $v$ is not covered & $\forall e \in N^e_{X_{i}}(v), e$ is covered & 0 
										\\
										\cline{2-6}
										6 & $v \notin \mathcal{MD}$ & $\forall e \in N^e_{X_{i}}(v), e \notin \mathcal{MD}$ & $v$ is covered & $\exists e  \in N^e_{X_{i}}(v), e$ is not covered & 0 
										\\
										\cline{2-6}
										7 & $v \notin \mathcal{MD}$ & $\forall e \in N^e_{X_{i}}(v), e \notin \mathcal{MD}$ & $v$ is not covered & $\exists e  \in N^e_{X_{i}}(v), e$ is not covered & 0 
										\\
										\cline{2-6}
										
									\end{tabular}
									
								}
					\end{table}
					
		These six states are the same as the seven states in
			Section~\ref{Sec3}, except for the state 2.
			The state 2 is deleted since it has the same effect as state 1 on $N^{md}_G[v]$, since we do not store any information on edges in the tables $Stable$ and $Btable$. 
			
	We follow the three phases of our proposed algorithm in Section~\ref{Sec4}. Note the storage tables 
		$Stable$ and $Btable$ store a valid assignment only for the vertices.
		During the traversal $\tau$, when we observe a \textsc{leaf} bag, a new table is created which saves all
		of the possible states for the only vertex in that bag. In an \textsc{introduce} or a \textsc{join} bag,
		the multiplication operators $\star_{Int}$ and $\star_{Join}$ in Equations~\ref{seq1} and~\ref{seq2},
		respectively, are used to
		compute the entries of $Stable_{i}$. This algorithm omits the invalid rows from $Stable_{i-1}$ to obtain
		the $Stable_{i}$ for bag  $X_{i}$ in a \textsc{forget} bag.
		
		It is clear that the time complexity of this algorithm relates to the time spent for processing each bag in
		the traversal of $\tau$. In a bag, we have at most $\tw + 1$ vertices. So, This	algorithm uses
		$O(6^{\tw})$ time on leaf, \textsc{introduce} and \textsc{forget} bags but $O(7^{\tw})$ time on a
		\textsc{join} bag. To reduce the time spent for a \textsc{join} bag, we use the fast subset convolution to
		multiply the two tables of size $6^{\tw}$ in time $O(6^{\tw})$. In this multiplication technique, we use
		the notion of fast subset convolution and convert the two tables $Stable$ and $Btable$ to two new tables
		$Stable'$ and $Btable'$. In these new tables, states $5$ and $7$ are merged to a new state $0_{?}$ where this
		vertex is not covered and its edges may or may not be covered. Therefore, it suffices to multiply the two
		tables since no states in a \textsc{join} bag tables are lost. Using the results of convolution and
		alternative representations of vertex states, our algorithm for MDS improves to $O^*(6^{\tw})$.
		
		\section{The Correctness of the Algorithms} \label{Sec5}
			In this section, we first show that finding the mixed domination number of graphs is checkable in linear-time if the graph has bounded tree-width using Courcelle's Theorem~\cite{courcelle1990monadic}. Then, we describe how to ensure that our proposed bottom-up method solves MDS by computing partial solutions as state tables for bags.  A partial solution is an object which stores all possible states for vertices and edges in a bag. Therefore, what we need to show is showing that how a partial solution can be extended to a final solution. 
				
				To show that the mixed domination property of graphs can be checked in linear-time for graphs with bounded tree-width, we consider the following to express mixed domination in monadic second-order logic. 
				\begin{itemize}
					\item Vertices, edges, sets of vertices and edges of a graph G as variables of monodic second order logic.
					\item Relations $adj(p, q)$ and $inc(p, q) $  which  are defined as follows: $adj(p, q)$ is a binary adjacency relation and it  is true if and only if $p$ and $ q$ are two adjacent vertices or are two adjacent edges of $G$, and $inc(p, q) $ is a binary incidence relation and it is true if and only if edge $p$ is incident to vertex $q$ or vertex  $p$ is incident to edge $q$.
					\item The set operations $\cup$, $\cap$, $\subseteq $ and $\in$  denote the union, the intersection, the
					subset and  the membership operators, respectively.
					\item Unary set cardinality operator $\left\vert{S}\right\vert$ and the set equality operator $=$.
					\item The comparison operator  $\leq$.
					\item The logical connectives AND $(\wedge)$ and  OR $(\vee)$.
					\item The logical quantifiers $\forall$ and $\exists$ over vertices, edges, sets of vertices and edges of G.
				\end{itemize}
				
				By modeling the mixed domination property of graphs in the described monadic second-order logic, we can conclude that checking this property on graphs with a bounded tree-width is a linear time task.
				Let $G=(V,E)$ be a given simple graph. For any element $r\in V \cup E$, the mixed neighborhood of $r$  is
				denoted by $N^{md}(r)$ and is defined as $N^{md}_G(r) = \{ s\in V \cup E \mid adj(s, r) \vee inc(s, r) \}$
				and the closed neighborhood of $r$ is denoted by $N^{md}_G[r]$ and equals $N^{md}_G[r] =N^{md}_G(r)\cup
				\{r\}$. A subset $\mathcal{S} \subseteq V \cup E$ is a mixed dominating set (MDS) for $G$, if for all ${r\in
					V \cup E},$ it is the case that $|N^{md}_G[r] \cap \mathcal{S}|\geq 1 $. The mixed domination number problem
				asks for the size of $S$ which is expressed as in Equation~\ref{mso}:
				
				\begin{equation}\label{mso}
				\exists S  \subseteq (V \cup E),\; S \text{ is an  MDS } \wedge\; \forall M\subseteq (V \cup E),\;  M \text{  is  an MDS},\; |S|  \leq |M|.
				\end{equation}
				
				To show the correctness of our algorithm, it is enough to show that the extension of a partial solution
				satisfies the condition of the restricted form of mixed domination problem. Extending a partial solution
				begins by constructing a table for a \textsc{leaf} bag. At first, the leaf table contains two possible
				states  (see Section~\ref{Create leaf}). Obviously, these two states are all possible states that can occur for a vertex in a bag.  
				
				According to the definition of the tree decomposition, a bag $X_{i}$ is a separator whenever it  separates the vertices of  $(X^{\triangle}_{i}) \backslash X_{i}$ from  $V\backslash (X^{\triangle}_{i}) $, so the vertices of $(X^{\triangle}_{i}) \backslash X_{i}$ do not appear in other bags except the ones descending from bag $X_i$. Hence, in our bottom-up approach,  all of the possible states have been considered for vertices of bag $X_i$ since they will never be considered again. Using this fact, we continue the proof of correctness of our algorithm by first checking the extension of an \textsc{introduce} bag. 
				
				\begin{lemma}\label{lemma1}
					Let $ Stable_{i-1}$ and $Btable_{i}$ be two tables with all possible states mentioned  in
					Table~\ref{condition vertex} and Table~\ref{bag condition vertex}  for  $X_{i-1}$ and $X_i$ respectively.
					Combining these two tables according to Algorithm~\ref{MixedAlgorithm} produces all of the  possible states for bag $X_i$.
				\end{lemma}
				\begin{proof}
					The proof is by contradiction. 
					A new states state $r\in Stable_{i}$ is the result of multiplying two possible states $r_1\in Stable_{i-1}$ and $r_2 \in Btable_i$. We need to show that state $r$ cannot be produced by impossible states.
					Let $r'_1$ be a possible and $r'_2$ be an impossible state, hence $r'_1\in Stable_{i-1}$ and $r'_2 \notin Btable_i$ and $r=r_1 \otimes r_2$.
					Since $r$ is a possible state, then it preserves all of the restrictions of MDS     while state
					$r'_2 \notin Btable_i$ does not satisfy those restrictions and is impossible. So, because we
					cannot satisfy some restrictions in the  entire problem while for a part of it, it is not
					satisfied. Similarly, the proof for cases $r'_1 \notin Stable_{i-1}$ and $r'_2 \in Btable_i$ or $r'_1 \notin Stable_{i-1}$ and $r'_2 \notin Btable_i$ is similar. Combining $Stable_{i-1} $ and $ Btable_{i}$ produces all of the possible states for bag  $X_i$.
				\end{proof}
				In the second step, we show how the extension of a \textsc{forget} bag satisfies restriction mixed domination
				problem. The proof is by contradiction. 
				
				\begin{lemma}\label{lemma2}
					Let $ Stable_{i-1}$ be a table with all possible states for  bag $X_{i-1}$. Deleting invalid rows from this table produces all of the possible states for bag $X_i$.
					
				\end{lemma}
				\begin{proof}
					Suppose that the vertex $v$ is the \textsc{forget} vertex, thus it does not appear in later steps, all of its edges have appeared up to now and all its possible states are checked before constructing $Stable_{i-1}$. For vertex $v$ and its edges, there are  four cases to consider:
					\begin{itemize}
						\item [1.] {\bf  vertex $v$ and its edges are dominated:} In these cases, the value of vertex $v$ is either 1, 2, 3 or 4. So, they  are valid cases and remain in $Stable_{i}$  
						\item [2.] {\bf  vertex $v$ is not dominated:} For these cases, the value of $v$ is either 5, 7 or 9. So, they are invalid cases and are not allocated to appear in $Stable_{i}$ since no neighbor of vertex $v$ appeared till now.
						\item [3.] {\bf vertex  $v$ is dominated but at least one of its edges in $G_{i-1}^{\triangle} \backslash G_{i-1}$ is not dominated:} Let $e$ be an edge in $G_{i-1}^{\triangle} \backslash G_{i-1}$ which is not dominated. Edge $e$ causes the value of vertex $v$ to be 8 and makes invalid cases. So, these cases must be omitted from $Stable_{i}$ since edge $e$ has not any neighbor which appeared till now.
						\item [4.] {\bf  vertex $v$ and its edges in $G_{i-1}^{\triangle} \backslash G_{i-1}$ are dominated
							but at least one of its edges in $ G_{i-1}$  is not dominated:} Let $e=(v,v')$ be an edge in $
						G_{i-1}$ that is not dominated. So, edge $e$ causes the value of $v$ to be 6. By deleting vertex $v$,
						edge $e$ is converted to an external edge for bag $i$ and does not appear in $ G_{i}$. However  it is
						possible that an edge appears in later steps and dominates $e$, so this states can be extended to a
						solution and we need to save it. Thus, we consider the effect of edge $e$ on vertex $v'$ and change the value of vertex $v'$. 
					\end{itemize}
					
				\end{proof}
				
				In the third step, we are going to check the correctness of our algorithm in the  extension of a \textsc{join} bag.  
				
				\begin{lemma}\label{lemma3}
					Let $ Stable_{i1}$ and $Stable_{i2}$ be tables with all possible states  mentioned in
					Table~\ref{condition vertex} and Table~\ref{bag condition vertex}  for children of bag $X_i$. Combining these  two tables  produces all of the possible states for bag $X_i$.
				\end{lemma}
				\begin{proof}
					The proof is  by contradiction and is similar to Lemma~\ref{lemma1}, except that   $r'_1\in Stable_{i1}$ and $r'_2 \notin Stable_{i2}$ and for other cases similarly $r'_1\notin Stable_{i1}$ and $r'_2 \in Stable_{i2}$ and $r'_1\notin Stable_{i1}$ and $r'_2 \notin Stable_{i2}$.
				\end{proof}
				
				\begin{theorem} \label{theorem2} Our proposed mixed domination algorithm \label{Mixed} produces all possible states  satisfying MDS.
				\end{theorem}
				\begin{proof}
					A partial solution starts by considering a \textsc{leaf} bag and initializing the leaf table to contain all
					two possible states. By Lemmas~\ref{lemma1},~\ref{lemma2} and~\ref{lemma3},  
					we proved that every extension of partial solutions contains all of the possible states that satisfy 
					the conditions of a mixed domination set,  Thus, our algorithm produces all possible states required.  
				\end{proof}
				\begin{lemma} \label{lemma5}
					Every extension of a partial solution preserves possible states with minimum cost.
				\end{lemma}
				
				\begin{proof}
					The Algorithm~\ref{MixedAlgorithm}  stores  rows of $Stable_{i}$ using a coding scheme and updates the state with minimum cost. So, every extension of a partial solution preserves possible states with minimum cost.   
				\end{proof}
				
				\begin{theorem} 
					Our proposed dynamic algorithm computes $\gamma_{md}$ of graphs with constant tree-width in the
					status table of root bag, e.g. $Stable_{root}$ in  $O(9^{\tw} \times  3^{\tw^{2}} \times  \tw^{2} \times   n)$ time. 
				\end{theorem}
				\begin{proof}
					According to the Theorem~\ref{theorem2} and Lemma~\ref{lemma5}, our algorithm preserves possible states with
					minimum cost in every bag, so $Stable_{root}$ contains all of the possible states with minimum cost. To find
					$\gamma_{md}$, it is enough for the algorithm to find the  smallest cost among all of the possible states in
					the root bag.  Also Theorem~\ref{timecom} expresses the time complexity of this algorithm. So, the desired result is obtained.
				\end{proof}
				
				\section{Conclusion} \label{Sec7}
				
				In this paper, we proposed an algorithm to solve MDS and AMDS of graphs with
				bounded tree-width.  Our algorithm used a novel technique of inputting edge
				domination power to vertices of the graph. Using this technique and by
				analyzing our algorithm, we have shown for the first time that MDS is
				fixed-parameter tractable.  We provided detailed dynamic program is given with 
				running time $O^*(3^{\tw^2})$ and a theoretical improved version with running time $O^*(6^\tw)$ . 
				
				As a future work, we suggest enhancing the running
				time of our algorithm for special classes of graphs. Studying other parameters
				for MDS such as path-width seems a fruitful topic, too. More importantly,
				exploiting our technique of assigning power values to vertices from edges to solve
				other graph problems is another research direction.
				
	\acknowledgements
	\label{sec:ack}
		This article has been written while the second author was on a sabbatical visit to the University of Auckland.
			He would like to express his gratitude to Prof.~Cristian S.~Calude and his research group for the nice and friendly hospitality.
			
			The revision of this article has been done when the first author was on research visit to Eotvos Lorand
			University, ELTE. He would like to express his thankfulness, warmth and appreciation to
			Prof.~Komjath P.~who made his research successful. He would also like to extend his thanks to the computer
			science group of the university of ELTE which assisted him at every point to cherish his goal.


\begin{thebibliography}{19}
\providecommand{\natexlab}[1]{#1}
\providecommand{\url}[1]{\texttt{#1}}
\expandafter\ifx\csname urlstyle\endcsname\relax
  \providecommand{\doi}[1]{doi: #1}\else
  \providecommand{\doi}{doi: \begingroup \urlstyle{rm}\Url}\fi

\bibitem[Adhar and Peng(1994)]{adhar1994mixed}
G.~S. Adhar and S.~Peng.
\newblock Mixed domination in trees: a parallel algorithm.
\newblock \emph{Congressus Numerantium}, 100:\penalty0 73--80, 1994.

\bibitem[Ahangar et~al.(2015{\natexlab{a}})Ahangar, Asgharsharghi,
  Sheikholeslami, and Volkmann]{ahangar2015signed}
H.~A. Ahangar, L.~Asgharsharghi, S.~Sheikholeslami, and L.~Volkmann.
\newblock Signed mixed roman domination numbers in graphs.
\newblock \emph{Journal of Combinatorial Optimization}, pages 1--19,
  2015{\natexlab{a}}.

\bibitem[Ahangar et~al.(2015{\natexlab{b}})Ahangar, Haynes, and
  Valenzuela-Tripodoro]{ahangar2015mixed}
H.~A. Ahangar, T.~W. Haynes, and J.~Valenzuela-Tripodoro.
\newblock Mixed roman domination in graphs.
\newblock \emph{Bulletin of the Malaysian Mathematical Sciences Society}, pages
  1--12, 2015{\natexlab{b}}.

\bibitem[Alavi et~al.(1977)Alavi, Behzad, Lesniak-Foster, and
  Nordhaus]{alavi1977total}
Y.~Alavi, M.~Behzad, L.~M. Lesniak-Foster, and E.~Nordhaus.
\newblock Total matchings and total coverings of graphs.
\newblock \emph{Journal of Graph Theory}, 1\penalty0 (2):\penalty0 135--140,
  1977.

\bibitem[Bodlaender(1996)]{bodlaender1996linear}
H.~L. Bodlaender.
\newblock A linear-time algorithm for finding tree-decompositions of small
  treewidth.
\newblock \emph{SIAM Journal on computing}, 25\penalty0 (6):\penalty0
  1305--1317, 1996.

\bibitem[Bodlaender and van Antwerpen-de
  Fluiter(2001)]{bodlaender2001reduction}
H.~L. Bodlaender and B.~van Antwerpen-de Fluiter.
\newblock Reduction algorithms for graphs of small treewidth.
\newblock \emph{Information and Computation}, 167\penalty0 (2):\penalty0
  86--119, 2001.

\bibitem[Chimani et~al.(2012)Chimani, Mutzel, and Zey]{chimani2012improved}
M.~Chimani, P.~Mutzel, and B.~Zey.
\newblock Improved steiner tree algorithms for bounded treewidth.
\newblock \emph{Journal of Discrete Algorithms}, 16:\penalty0 67--78, 2012.

\bibitem[Courcelle(1990)]{courcelle1990monadic}
B.~Courcelle.
\newblock The monadic second-order logic of graphs. {I}. recognizable sets of
  finite graphs.
\newblock \emph{Information and computation}, 85\penalty0 (1):\penalty0 12--75,
  1990.

\bibitem[Courcelle(1992)]{courcelle1992monadic}
B.~Courcelle.
\newblock The monadic second-order logic of graphs iii: Tree-decompositions,
  minors and complexity issues.
\newblock \emph{RAIRO-Theoretical Informatics and Applications}, 26\penalty0
  (3):\penalty0 257--286, 1992.

\bibitem[Courcelle(2015)]{courcelle2015fly}
B.~Courcelle.
\newblock Fly-automata for checking monadic second-order properties of graphs
  of bounded tree-width.
\newblock \emph{Electronic Notes in Discrete Mathematics}, 50:\penalty0 3--8,
  2015.

\bibitem[Hatami(2007)]{hatami2007approximation}
P.~Hatami.
\newblock An approximation algorithm for the total covering problem.
\newblock \emph{Discussiones Mathematicae Graph Theory}, 27\penalty0
  (3):\penalty0 553--558, 2007.

\bibitem[Haynes et~al.(1998)Haynes, Hedetniemi, and
  Slater]{haynes1998fundamentals}
T.~W. Haynes, S.~Hedetniemi, and P.~Slater.
\newblock \emph{Fundamentals of domination in graphs}.
\newblock CRC Press, 1998.

\bibitem[Jain et~al.(2017)Jain, Jayakrishnan, Panolan, and Sahu]{jain2017m}
P.~Jain, M.~Jayakrishnan, F.~Panolan, and A.~Sahu.
\newblock Mixed dominating set: A parameterized perspective.
\newblock In \emph{International Workshop on Graph-Theoretic Concepts in
  Computer Science}, pages 330--343. Springer, 2017.

\bibitem[Lan and Chang(2013)]{lan2013mixed}
J.~K. Lan and G.~J. Chang.
\newblock On the mixed domination problem in graphs.
\newblock \emph{Theoretical Computer Science}, 476:\penalty0 84--93, 2013.

\bibitem[Rajaati et~al.(2016)Rajaati, Hooshmandasl, Dinneen, and
  Shakiba]{rajaati2016fixed}
M.~Rajaati, M.~R. Hooshmandasl, M.~J. Dinneen, and A.~Shakiba.
\newblock On fixed-parameter tractability of the mixed domination problem for
  graphs with bounded tree-width.
\newblock \emph{arXiv preprint arXiv:1612.08234}, 2016.

\bibitem[Robertson and Seymour(1984)]{robertson1984graph}
N.~Robertson and P.~D. Seymour.
\newblock Graph minors. {III}. planar tree-width.
\newblock \emph{Journal of Combinatorial Theory, Series B}, 36\penalty0
  (1):\penalty0 49--64, 1984.

\bibitem[Van~Rooij et~al.(2009)Van~Rooij, Bodlaender, and
  Rossmanith]{van2009dynamic}
J.~M. Van~Rooij, H.~L. Bodlaender, and P.~Rossmanith.
\newblock Dynamic programming on tree decompositions using generalised fast
  subset convolution.
\newblock In \emph{Algorithms-ESA 2009}, pages 566--577. Springer, 2009.

\bibitem[West et~al.(2001)]{west2001introduction}
D.~B. West et~al.
\newblock \emph{Introduction to graph theory}, volume~2.
\newblock Prentice hall Upper Saddle River, 2001.

\bibitem[Zhao et~al.(2011)Zhao, Kang, and Sohn]{zhao2011algorithmic}
Y.~Zhao, L.~Kang, and M.~Y. Sohn.
\newblock The algorithmic complexity of mixed domination in graphs.
\newblock \emph{Theoretical Computer Science}, 412\penalty0 (22):\penalty0
  2387--2392, 2011.

\end{thebibliography}

\label{sec:biblio}

\end{document}